\documentclass[conference]{IEEEtran}
\usepackage{hyperref}
\usepackage{quiver}
\usepackage{macros}
\usepackage[style=ieee]{biblatex}
\IEEEoverridecommandlockouts
\begin{document}

\title{Thin Coalgebraic Behaviours Are Inductive}

\author{\IEEEauthorblockN{Anton Chernev}
\IEEEauthorblockA{
\textit{University of Groningen}\\
Groningen, Netherlands \\
a.chernev@rug.nl}
\and
\IEEEauthorblockN{Corina C\^{i}rstea}
\IEEEauthorblockA{\textit{University of Southampton}\\
Southampton, United Kingdom \\
cc2@soton.ac.uk}
\and
\IEEEauthorblockN{Helle Hvid Hansen}
\IEEEauthorblockA{
\textit{University of Groningen}\\
Groningen, Netherlands \\
h.h.hansen@rug.nl}
\and
\IEEEauthorblockN{Clemens Kupke}
\IEEEauthorblockA{\textit{University of Strathclyde}\\
Glasgow, United Kingdom \\
clemens.kupke@strath.ac.uk}
\thanks{C\^{i}rstea and Kupke were funded by a Leverhulme Trust Research Project Grant (RPG-2020-232).}
}

\maketitle

\begin{abstract}
Coalgebras for analytic functors uniformly model graph-like systems where the successors of a state may admit certain symmetries. Examples of successor structure include ordered tuples, cyclic lists and multisets. Motivated by goals in automata-based verification and results on thin trees, we introduce thin coalgebras as those coalgebras with only countably many infinite paths from each state. Our main result is an inductive characterisation of thinness via an initial algebra. To this end, we develop a syntax for thin behaviours and capture with a single equation when two terms represent the same thin behaviour. Finally, for the special case of polynomial functors, we retrieve from our syntax the notion of Cantor-Bendixson rank of a thin tree.
\end{abstract}


\section{Introduction}

\paragraph*{Background and motivation}
Coalgebra~\cite{Rutten:TCS2000,Jacobs2016book} is a well-established categorical framework for modelling and reasoning about a wide variety of state-based systems.
Coalgebras are defined for an endofunctor $F$, which specifies the system type, 
and this abstraction step has proved useful for developing a universal theory of systems parametric in $F$. For example, program semantics \cite{GoncharovMilTsaUrb2024,RotBonchiBonsanguePRS2017,JacobsSilvaSokolova2015,BaldanBonchiKerstanKoenig2018}, 
logics \cite{KupkePattinson2011,CirsteaKupkePattinson2011}, 
automata theory \cite{Rutten99,VenemaKupke:CoalgebraicAutomata} and 
verification techniques \cite{CirsteaShimizuHasuo,HausmannHPSS24} can be developed uniformly for a large variety of system types, including some that are not covered by existing approaches. 
Recent work in this direction \cite{CirsteaKupke:CSL2023} has shown that automata-based verification generalises smoothly to a large class of coalgebraic models, provided that the automata used to capture correctness properties are assumed to be unambiguous. However, just like automata on infinite trees, unambiguous coalgebra automata are less expressive than their non-deterministic counterparts. For automata on infinite trees, one way to regain expressive power is to restrict the input to \emph{thin trees}~\cite{BojIdzSkr:STACS13,IdzSkrBoj:2016}, that is, infinite trees with only countably many infinite branches. They have a (transfinite) inductive structure~\cite[Sec.~6.1.4]{Skrzypczak2016}, which facilitates a well-behaved language theory, closer to the theory of infinite words. In particular, regular languages of thin trees can be unambiguously accepted~\cite[Thm.~32]{BilkowskiSkrzypczak:CSL13}, \cite[Thm.~12]{BojIdzSkr:STACS13}.

Driven by the aim to use coalgebra automata for verification, we ask the question: 
\emph{To what extent can results for thin trees be generalised to structures beyond trees, using coalgebra?}

To answer this question, we must find a sweet spot between a high level of generality for the types of structures we consider, and ensuring that key properties such as
admitting an inductive structure
(crucial for inheriting the tractability of thin trees)
are maintained. From a modelling perspective, we are interested in structures that describe runs of a state-based system; these include infinite words and infinite trees,
but of interest are also graph-like structures where successor states are organised according to an abstract data type.
The latter can model multi-process systems, e.g., a server spawning multiple subprocesses, partially ordered by priority.

\paragraph*{Contributions}
In this paper, we define and study thin coalgebras of \emph{analytic functors} \cite{JoyalFoncteursAnalytiques} (see also \cite{HasegawaAnalyticFunctors}).
We summarise our main contributions:

1) We identify analytic functors as a suitable restriction on coalgebra types for which a theory of thin structures can be developed.
An analytic functor specifies a type of successor structure that may admit certain symmetries.
At one extreme, polynomial functors describe structures where there are no symmetries governing the successors of a state. 
In particular, ranked, ordered trees are coalgebras for a polynomial functor. 
At the other extreme is the bag functor, whose coalgebras are unordered multigraphs.
In between, one finds, for example,
the type of cyclic lists (lists that can be shifted cyclically), or the type of posets (as in our previous server example).
Coalgebras for analytic functors thus capture a wide variety of graph-like structures.
At the same time, they support a generic notion of \emph{infinite path} 
(generalising the notion of infinite branch in a tree) 
and crucially, the number of infinite paths is invariant under coalgebra morphisms.

2) For an analytic functor $F$, 
we define a notion of thin state in an $F$-coalgebra.
Informally, a state is \emph{thin} if there are countably many infinite paths from it.
This yields a notion of \emph{thin behaviour} as a thin state in the final coalgebra.
Thin coalgebra states generalise thin trees in two ways: to coalgebras (which, unlike trees, may contain cycles -- a feature that allows, e.g., the finite representation of regular trees), and to more general transition types (trees are special coalgebras for polynomial functors). 
We also provide a criterion for thinness which can be verified in linear time (\Cref{prop:thin-cycles}). 

3) Given an analytic functor $F$, we define a syntax for (thin) behaviours 
as the initial algebra for the functor $F+G$ where $G = (F'-)^\omega$ is defined via the functor derivative $F'$ \cite{AbbottEtAl:DifferentiatingDataStructures}.
This syntax is equipped with equivalent operational and denotational semantics as a map $\interpr-$ into the final $F$-coalgebra.
We call a behaviour \emph{constructible} if it is in the image of $\interpr-$. 
We axiomatise with a single equation when two terms represent the same behaviour (Soundness, \Cref{thm:soundnessOfQuot} and Completeness, \Cref{thm:completeness}).
An $(F+G)$-algebra that satisfies the equation is called \emph{coherent}. 
To obtain our inductive characterisation of thinness, we show that thin behaviours are precisely the constructible ones (\Cref{thm:thinIFFConstructible}) and that constructible behaviours form an initial coherent algebra (\Cref{thm:thinImageAndQuotientIsomorphic}).

4) We introduce \emph{normal terms} as canonical representatives for thin behaviours,
and use them to assign to each thin behaviour an ordinal rank. 
We show that for polynomial functors $F$, this rank captures the notion of Cantor-Bendixson rank of thin trees from descriptive set theory (\Cref{thm:CB}). 
Thus, our ranks can be seen as providing a measure of thinness.
Moreover, normal terms are instrumental in proving soundness and completeness (\Cref{thm:soundnessOfQuot,thm:completeness}).

\paragraph*{Related Work}

We briefly discuss how our results relate to similar results for thin trees.
To our knowledge, thin trees have only been studied in the setting of 
ordered trees \cite{BojIdzSkr:STACS13,IdzSkrBoj:2016,Skrzypczak2016book,BilkowskiSkrzypczak:CSL13}.  
As already mentioned, thin coalgebras for analytic functors are a strict generalisation of thin trees as they allow for
a wide range of successor types (the type of trees is a special case), as well as structures with cycles.
We also note that the generated behaviour of a state cannot, in general, be seen as a tree due to the symmetries in the successor structure. Behaviour could be represented as some equivalence class of trees, but that is cumbersome to work with.

Thin trees have been characterised algebraically via thin algebras \cite{Skrzypczak2016,BojIdzSkr:STACS13},
which are two-sorted generalisations of $\omega$-semigroups and Wilke algebras, with one sort for trees and the other for contexts of arbitrary depth. 
Our coherent algebras have only one sort (for behaviours) and 
use only contexts of depth $1$, modelled by the functor derivative. 
This one-step structure is the basis for the coalgebraic semantics and the axiomatisation with a single equation, which generalises one of the $\omega$-semigroup axioms. 

The terms of our syntax show some similarities with the notion of skeleton from op. cit.
For example, a tree is thin iff it has a skeleton. However, there are also notable differences.
A skeleton of a tree is a subset of nodes (that satisfies certain conditions) which has no internal structure,
whereas our terms are structured and capture the Cantor-Bendixson rank.
Moreover, the canonical skeleton of a tree relies on the arbitrary choice of always going to the first child,
whereas our normal terms are defined canonically and uniformly in the functor~$F$.

Finally, we mention related work on analytic functors as a basis for specifying \emph{abstract data types} \cite{AbbottEtAl2004QuotientDatatypes,CaretteUszkaySpeciesMakingAnalyticFunctorsPractical}.
In this context, coalgebras for analytic functors provide semantics for coinductive types, so our inductive characterisation of thin coalgebras can be interpreted as: thin $F$-behaviours are an inductive subtype of the coinductive type of all $F$-behaviours (assuming that the type system supports streams).

\smallskip
We include an appendix with omitted proofs.


\section{Preliminaries}\label{sec:preliminaries}

We assume familiarity with basic category theory, see e.g. \cite{Awodey_2006,AdaHerStr:JoyOfCats,MacLane_1971}. 
We denote with $\Set$ the category of sets and functions.
Given sets $X, Y \in \Set$, we write $Y^X$ for the set of functions $\phi: X \to Y$. We write $X + Y$ for the coproduct of $X$ and $Y$, and let $\inj[X+Y]{1}: X \to X+Y$ and $\inj[X+Y]{2}: Y \to X+Y$ denote the coproduct injections. For coproducts over an arbitrary index set $I$, we write $\bigsqcup_{i \in I}X_i$ and let $(i, x)$ with $i \in I$ and $x \in X_i$ denote an arbitrary element of the coproduct.

Let $F$ be a $\Set$-endofunctor, i.e., $F\colon \Set \to\Set$. 
An \emph{$F$-algebra} is a pair $(C, \gamma)$ where $C$ is a set and $\gamma: FC \to C$ is a function, called the \emph{algebra structure}. A morphism between two $F$-algebras $(C,\gamma)$ and $(D,\delta)$ is a map $f: C \to D$ such that $f \circ \gamma = \delta \circ Ff$. The category of $F$-algebras and $F$-algebra morphisms is denoted by $\Alg(F)$. An \emph{initial $F$-algebra} is an $F$-algebra $(A,\alpha)$ such that for all $F$-algebras $(C,\gamma)$ there is a unique $F$-algebra morphism $\ev \colon (A,\alpha) \to (C,\gamma)$. The algebra structure of an initial $F$-algebra is an isomorphism. An initial algebra can be seen as an algebra of terms and the initial morphism $\ev$ evaluates terms in $(C,\gamma)$.

The dual notion of algebra is \emph{coalgebra} \cite{Rutten:TCS2000}.
An \emph{$F$-coalgebra} is a pair $(X,\xi)$ where $X$ is a set and $\xi: X \to FX$ is a function, called the \emph{coalgebra structure}. A morphism between two $F$-coalgebras $(X,\xi)$ and $(Y,\delta)$ is a map $f: X \to Y$ such that $Ff \circ \xi = \delta \circ f$. A \emph{final $F$-coalgebra} is an $F$-coalgebra $(Z,\zeta)$ such that for all $F$-coalgebras $(X,\xi)$ there is a unique $F$-coalgebra morphism $\beh \colon (X,\xi) \to (Z,\zeta)$. The final coalgebra structure $\zeta$ is again an isomorphism. $F$-coalgebras can be seen as state-based systems, and a final $F$-coalgebra can then be seen as a domain of abstract, observable behaviours. The final morphism $\beh$ maps a state to its behaviour.
A classic example of a final coalgebra, which will also be used in this paper, is given by the set  $X^\omega$ of \emph{streams} over a set $X$, which forms the carrier of a final $(X \times \Id)$-coalgebra. The coalgebra structure on $X^\omega$ is given by the head and tail maps $\langle \head,\tail \rangle : X^\omega \to X \times X^\omega$.
The map $(-)^\omega$ can be made into a $\Set$-functor by defining $f^\omega(x_0, x_1, \ldots) \coloneqq (f(x_0), f(x_1), \ldots)$.

We will work with the factorisation system $(\cal E,\cal M)$ for the category $\Set$, where $\cal E$ consists of all epis and $\cal M$ consists of all monos. In $\Set$ these are precisely the surjective and injective functions, respectively.
This yields a factorisation system for the category $\Alg(F)$ consisting of the surjective and injective morphisms
(since all $\Set$-functors preserve epis), see e.g.~\cite{wissmann:CALCO2021:MinimalityNotionsViaFactorizationSystems}. 
Given an $F$-algebra morphism $f: (C, \gamma) \to (D,\delta)$ and its factorisation
\begin{tikzcd}
	{(C,\gamma)} & {(E,\epsilon)} & {(D,\delta)}
	\arrow["e", two heads, from=1-1, to=1-2]
	\arrow["m", tail, from=1-2, to=1-3]
\end{tikzcd}\, we have that $(E,\epsilon)$ is isomorphic to the subalgebra of $(D, \delta)$ with carrier $\Im(f)$.

In this paper, we work with $F$-coalgebras for \textit{analytic functors} \cite{JoyalFoncteursAnalytiques} (see also \cite{HasegawaAnalyticFunctors}). Analytic functors were introduced in the context of enumerative combinatorics to give a foundation to generating functions. In computer science, they serve as a formalisation of data types with symmetries \cite{Yorgey2010SpeciesAndFunctorsAndTypes,AbbottEtAl2004QuotientDatatypes}.

Before defining analytic functors, we recall basics of \textit{permutation groups}. Given a set $U$, let $\Sym(U)$ denote the group of permutations over $U$, i.e., bijections $\sigma: U \to U$. Subgroups of $\Sym(U)$ are called permutation groups. Given sets $U, X$ and a subgroup $H \leq \Sym(U)$, $H$ acts on $X^U$ by $\sigma \cdot \phi \coloneqq \phi \circ \sigma^{-1}$, for $\sigma \in H, \phi \in X^U$. We write $X^U / H$ for the set of orbits of the action of $H$ on $X^U$, where an orbit is of the form $[\phi]_H = \{ \psi \in X^U \mid \exists \sigma \in H : \psi = \sigma \cdot \phi \}$.

\begin{definition}[Analytic functor]
    An \emph{analytic functor} is a functor $F: \Set \to \Set$ of the form:
    \begin{equation*}
        F(X) = \bigsqcup_{i \in I} X^{U_i} / H_i, \quad F(f)([\phi]_{H_i}) = [f \circ \phi]_{H_i},
    \end{equation*}
    where, for every $i \in I$, $U_i$ is finite and $H_i \leq \Sym(U_i)$.
\end{definition}

\begin{remark}
    Results in this paper hold even when $U_i$ are countable, but we keep to the standard definition for clarity.
\end{remark}

In the above definition, $F(f)$ is well-defined, because if $[\psi]_H = [\phi]_H$ (witnessed by $\sigma \in H$), then $[f \circ \phi]_H = [f \circ \psi]_H$ (witnessed by $\sigma^{-1}$).

By requiring all $H_i$ to be the trivial group, one obtains the class of \emph{polynomial functors}. A polynomial functor corresponds to an algebraic signature $I$ where $i \in I$ is an operation symbol of arity $n_i = |U_i|$.
Given a polynomial functor $F$, elements of the final $F$-coalgebra can be seen as ranked ordered trees, called \emph{$F$-trees} \cite{AdamekPorst:OnTreeCoalgebras}. An $F$-tree $t$ consists of a root, labelled by some $i \in I$, and $n_i$-many immediate subtrees $t_0, \dotsc, t_{n_i-1}$. The final coalgebra structure maps $t$ to $(i, \phi)$, where $\phi(j) = t_j$ for all $0 \leq j < n_i$. We discuss $F$-trees further in \Cref{sec:thin-trees-classic}.

\begin{example}
\label{ex:polynomialFunctor}
    As a concrete example of a polynomial functor, consider $F(X) \coloneqq X + X^2$, i.e., $F$ corresponds to a signature with a unary operation symbol $\mathit{op}_1$ and a binary operation symbol $\mathit{op}_2$.
    In an $F$-tree, a node labelled with $\mathit{op}_1$ has one child, and a node labelled with $\mathit{op}_2$ has two children.
\end{example}

\begin{example}
    The \emph{bag functor} $\Bag (X) \coloneqq \bigsqcup_{n\in\omega}X^{n} / \Sym(n)$ is a well-known example of an analytic functor. An element $(n, [\phi]) \in \Bag(X)$ can be identified with a multiset of cardinality $n$. Behaviours for the bag functor can be seen as unordered multi-trees, where each edge has a multiplicity, but the order of successors does not matter.
\end{example}

\begin{example}
    Polynomial functors and the bag functor can be seen as two extremes, with the former performing no quotienting and the latter performing complete quotienting. An example of a functor in between is the type of cyclic lists, i.e., lists without a fixed initial index. This functor  can be written as $\mathcal{C}(X) = \bigsqcup_{n\in\omega}X^{n} / H_n$, where $H_n$ is generated by the permutation $\sigma_n(i) \coloneqq i + 1 \: (\text{mod } n)$.
\end{example}

The \emph{functor derivative}~\cite{AbbottEtAl:DifferentiatingDataStructures} $F'$ captures the type of one-hole contexts over $F$.
Informally, elements of $F'X$ are elements of $FX$ where one piece of data is replaced by a hole. 
To make this precise, 
we extend the action of a permutation group to the set $\bigsqcup_{u \in U} X^{U \setminus \{u\}}$ of partial functions $U \rightharpoonup X$ that are undefined precisely on one element of $U$. Given $\sigma \in H \leq \Sym(U)$, $u \in U$ and $\phi: U \setminus \{u\} \to X$, we define $\sigma \cdot (u, \phi) \coloneqq (\sigma(u), \phi \circ (\sigma^{-1}|_{U \setminus \{\sigma(u)\}}))$. We write an orbit of this action as $[u, \phi]_H$, for $(u, \phi) \in \bigsqcup_{u \in U}X^{U \setminus \{u\}}$.

\begin{definition}[Functor derivative]
    Given an analytic functor $F = \bigsqcup_{i \in I} (-)^{U_i}/H_i$, we define its \emph{functor derivative}:
    \begin{equation*}
        F' = \bigsqcup_{i \in I}\Big(\bigsqcup_{u \in U_i}(-)^{U_i \setminus \{u\}}\Big) / H_i.
    \end{equation*}
    Elements of $F'(X)$, called \emph{($F$-)contexts}, are pairs $(i, [u,\phi]_{H_i})$, where $i \in I$, $u \in U_i$, $\phi: U_i \setminus \{u\} \to X$.
    Define the associated \emph{plug-in} natural transformation 
    $\trig: F' \times \Id \Rightarrow F$ by $\trig_X((i, [u,\phi]_{H_i}), x) \coloneqq (i, [\phi \cup \{ \langle u, x \rangle \}]_{H_i})$.
\end{definition}

The plug-in takes a context and an element and fills the hole in the context with the element. See \cite{AbbottEtAl:DifferentiatingDataStructures} for a detailed discussion of functor derivatives.

\begin{example}
\label{ex:polynomialFunctorDerivative}
    For the polynomial functor $F(X) \coloneqq X + X^2$, its derivative is $F'(X) = X^0 + X^1 + X^1 \cong 1 + 2 \times X$. An $F$-context is
    either a node labelled $\mathit{op}_1$ with a hole as its only child,
    or it is a node labelled $\mathit{op}_2$ with two children, one of which is a hole. In the latter case, the index of the hole is given by an element of $2 = \{ 0, 1 \}$.
\end{example}

\begin{example}
    For the bag functor $\Bag$, we have $\Bag' \cong \Bag$, because a one-hole bag of size $n$ is simply a bag of size $n - 1$.
\end{example}

We will make use of the following nice properties. An analytic functor $F$ preserves inclusions and intersections. Moreover, $F$ is bounded (hence accessible), so the initial $F$-algebra and the final $F$-coalgebra exist~\cite[Theorem~6.10]{AdamekMiliusMoss2018FixedPointsOfFunctors}.
Two immediate propositions about the functor derivative and the plug-in are stated next. These properties make use of the notion of base for an intersection-preserving endofunctor (cf.~\cite[Theorem~8.1]{Gumm2005FromTCoalgebrasToFilterStructures}).

\begin{definition}
    Assume $F : \Set \to \Set$ preserves intersections. For $X \in \Set$, the \emph{base of $y \in F X$} is given by $\Base_F(y) \coloneqq \bigcap_{X' \subseteq X, y \in F X'}X'$. In particular, $\Base_F(y) \subseteq X$.
\end{definition}
Under the assumption that $F$ is analytic (and hence intersection-preserving), the above instantiates to $\Base_F(i,[\phi]_{H_i}) = \Im(\phi)$ for $i \in I$ and $\phi \in X^{U_i}$. Similarly, $\Base_{F'}(i,[u,\phi]_{H_i}) = \Im(\phi)$ for every context $(i,[u,\phi]_{H_i})$. Hence bases for $F$ and $F'$ are finite. The next proposition expresses useful elementary properties of the plug-in.
\begin{apxpropositionrep}
\label{prop:elemPropertiesOfDerivative}
    Let $F$ be an analytic functor and $X$ a set.
    \begin{enumerate}[(i)]
        \item\label{prop:baseOfPlugin} If $x \in X$, $\bar x' \in F'X$, then $\Base_F(\trig_X(\bar x', x)) = \Base_{F'}(\bar x') \cup \{ x \}$.
        \item If $\bar x \in FX$ and $x \in \Base_F(\bar x)$, then there exists $\bar x' \in F'X$ with $\trig_X(\bar x', x) = \bar x$.
        \item If $\bar x', \bar y' \in F'X$, $x \in X \setminus \Base_{F'}(\bar x')$ and $\trig_X(\bar x', x) = \trig_X(\bar y', x)$, then $\bar x' = \bar y'$.
    \end{enumerate}
\end{apxpropositionrep}
\begin{proof}
    \textbf{(i).} Suppose $\bar x' = (i, [u, \phi]_{H_i})$ for some $i \in I$, $u \in U_i$, $\phi: U_i \setminus \{u\} \to X$ and let $\psi \coloneqq \phi \cup \{\langle u, x \rangle\}$. Then:
    \begin{multline*}
        \Base_F(\trig_X(\bar x', x)) = \Base_F(i, [\psi]_{H_i}) = \Im(\psi) = \\
        \Im(\phi) \cup \{ x \} = \Base_{F'}(\bar x') \cup \{ x \}.
    \end{multline*}

    \textbf{(ii).} Suppose $\bar x = (i, [\phi]_{H_i})$. The condition $x \in \Base(\bar x)$ means $x \in \Im(\phi)$. Let $u$ be any element in the non-empty set $\phi^{-1}(\{x\})$ and define $\psi =\phi \setminus \{ \langle u, x \rangle \}$. Then the context $\bar x' \coloneqq (i, [u, \psi]_{H_i})$ satisfies the condition $\trig_X(\bar x', x) = \bar x$.

    \textbf{(iii).} Suppose $\trig_X(\bar x', x) = \trig_X(\bar y', x) = (i, [\phi]_{H_i}) \in FX$. Then $\bar x' = (i, [u, \psi]_{H_i})$ for some $u \in U_i$ and $\psi: U_i \setminus \{ u \} \to X$ such that $\psi \cup \{\tup{u, x}\} = \phi$. Since $x \notin \Base_{F'}(\bar x') = \Im(\psi)$, it follows that $\phi^{-1}(\{x\}) = \{u\}$ and $\psi = \phi \setminus \tup{u,x}$,
    so $\bar x'$ is uniquely determined by $\phi$. Analogously, $\bar y'$ is uniquely determined by $\phi$, and thus $\bar x' = \bar y'$.
\end{proof}

\begin{assumption}
    For the remainder of this paper, we fix an analytic functor $F = \bigsqcup_{i \in I} (-)^{U_i} / H_i$.
\end{assumption}


\section{Thin coalgebras}
\label{sec:thin-coalgebras}

In this section, we introduce the central notion of the paper: thin coalgebras for analytic functors. We take a combinatorial perspective and define thin coalgebras via infinite paths, which allows for graph-theoretic reasoning. 
We begin with a formal definition of paths in an $F$-coalgebra.

\begin{definition}[Successor]
    Let $(T,\tau)$ be an $F$-coalgebra (for an analytic functor). Given $t, t' \in T$ with $\tau(t) = (i, [\phi]_{H_i})$ and $t' \in \Im(\phi)$, we say $t'$ is a \emph{successor} of $t$ with \emph{multiplicity} $|\phi^{-1}(t')|$. We write $\Suc(t) \coloneqq \{ (t', k') \mid \text{$t'$ is a successor of $t$ with multiplicity $l > k'$} \}$.
\end{definition}

The notion of successor with multiplicity implicitly defines a multigraph on $T$, which we refer to as the \emph{successor-multigraph} of $(T,\tau)$. 
The next definition of infinite path through a coalgebra allows us to distinguish between different paths in the successor-multigraph.
\begin{definition}[Path]
    Let $(T, \tau)$ be an $F$-coalgebra and $t_0 \in T$. An \emph{infinite path} from $t_0$ is an infinite sequence $(t_0k_1t_1k_2t_2 \dotsc) \in (T \times \omega)^\omega$ where $(t_{i+1}, k_{i+1}) \in \Suc(t)$ for every $i\in \omega$. We write $\InfPath(t_0)$ for the set of infinite paths from $t_0$. A \emph{finite path} from $t_0$ to $t_n$ is a finite sequence $t_0k_1t_1 \dotsc k_nt_n \in T \times (\omega \times T)^n$ where $(t_{i+1}, k_{i+1}) \in \Suc(t_i)$ for every $i < n$.
\end{definition}

Hence (finite/infinite) paths refer to sequences of states with additional information to account for the different ways to transition from one state to another.
Note that it is important for the definition of path to be independent of the choices of representatives $\phi \in [\phi]_{H_i}$. 
Thus paths do not record indices $u \in U_i$, which are dependent on the choice of representative.

\begin{definition}[Thin coalgebra]
\label{def:thinCoalgebra}
    Let $(T, \tau)$ be an $F$-coalgebra. An element $t \in T$ is \emph{thin} if there are at most countably many infinite paths from $t$. The coalgebra $(T,\tau)$ a \emph{thin} if all its elements are thin.
\end{definition}

\begin{example}
    Consider the analytic functor $F(X) = 1 + X + X^3 / H$, where $H$ is generated by $\sigma \in \Sym(3)$ with $\sigma(0) \coloneqq 0, \sigma(1) \coloneqq 2, \sigma(2) \coloneqq 1$. A state in an $F$-coalgebra has either no successors, one successor, or three successors with one of them ``marked''.
    An example of a thin $F$-coalgebra is depicted in \Cref{fig:thinCoalgThin}. It represents the execution of a server $s$ which, at every step, spawns two worker processes $w_1$ and $w_2$ and returns to itself. The marked successor, designated by a dashed arrow, is the server. An infinite path from $s$ is either equal to $s(0s)^\omega$ or of the form $s(0s)^n(0w_1)^\omega$. 
    Hence there are countably many of them and so $s$ is thin. The other states $w_1$ and $w_2$ are also thin, so the whole coalgebra is thin.
\end{example}

\begin{example}
    In \Cref{fig:thinCoalgNonThin} we see an example of a non-thin coalgebra for the bag functor $\Bag$. The transition at $s$ is a bag that contains two copies of $s$. Given an infinite sequence $(k_n)_{n\in\omega} \in \{ 0, 1 \}^\omega$, we have that $s k_0 s k_1 s \dotsc$ is a path. Hence there are uncountably many paths from $s$. This coalgebra is behaviourally equivalent to the full binary tree (\Cref{fig:thinCoalgBinaryTree}), which is a canonical example of a non-thin tree~\cite{Skrzypczak2016}.
\end{example}

\begin{figure}
    \centering
    \begin{subfigure}[b]{0.13\textwidth}
        \centering
        \scalebox{0.8}{
        \begin{tikzcd}[ampersand replacement=\&]
            \& {w_1} \\
            s \& {w_2}
            \arrow[from=1-2, to=1-2, loop, in=55, out=125, distance=10mm]
            \arrow[from=2-1, to=1-2]
            \arrow[dashed, from=2-1, to=2-1, loop, in=145, out=215, distance=10mm]
            \arrow[from=2-1, to=2-2]
        \end{tikzcd}
        }
        \caption{Thin}
        \label{fig:thinCoalgThin}
    \end{subfigure}
    \begin{subfigure}[b]{0.17\textwidth}
        \centering
        \scalebox{0.8}{
        \begin{tikzcd}
            s
            \arrow[from=1-1, to=1-1, loop, in=145, out=215, distance=10mm]
            \arrow[from=1-1, to=1-1, loop, in=35, out=325, distance=10mm]
        \end{tikzcd}
        }
        \caption{Non-thin}
        \label{fig:thinCoalgNonThin}
    \end{subfigure}
    \begin{subfigure}[b]{0.10\textwidth}
        \scalebox{0.8}{
        \begin{tikzpicture}
            \filldraw (6, 9) circle (1.5pt);
            \filldraw (5.5, 8.5) circle (1.5pt);
            \filldraw (6.5, 8.5) circle (1.5pt);
            \filldraw (5.25, 8.) circle (1.5pt);
            \filldraw (5.75, 8.) circle (1.5pt);
            \filldraw (6.25, 8.) circle (1.5pt);
            \filldraw (6.75, 8.) circle (1.5pt);
            \draw (6, 9) -- (5.5, 8.5);
            \draw (6, 9) -- (6.5, 8.5);
            \draw (5.5, 8.5) -- (5.25, 8.0);
            \draw (5.5, 8.5) -- (5.75, 8.0);
            \draw (6.5, 8.5) -- (6.25, 8.0);
            \draw (6.5, 8.5) -- (6.75, 8.0);
            \draw (5.25, 8.0) -- (5.1, 7.7);
            \draw (5.25, 8.0) -- (5.4, 7.7);
            \draw (5.75, 8.0) -- (5.6, 7.7);
            \draw (5.75, 8.0) -- (5.9, 7.7);
            \draw (6.25, 8.0) -- (6.1, 7.7);
            \draw (6.25, 8.0) -- (6.4, 7.7);
            \draw (6.75, 8.0) -- (6.6, 7.7);
            \draw (6.75, 8.0) -- (6.9, 7.7);
            \node at (5.25, 7.5) {$\vdots$};
            \node at (5.75, 7.5) {$\vdots$};
            \node at (6.25, 7.5) {$\vdots$};
            \node at (6.75, 7.5) {$\vdots$};
        \end{tikzpicture}
        }
        \caption{Non-thin}
        \label{fig:thinCoalgBinaryTree}
    \end{subfigure}
    \caption{Examples of thin and non-thin coalgebras}
    \label{fig:thinCoalg}
\end{figure}

The following proposition expresses that the number of infinite paths from a state is invariant under coalgebra morphisms. 

\begin{apxpropositionrep}
\label{prop:morphismsPreservePathCount}
    If $f: (T,\tau) \to (S,\sigma)$ is an $F$-coalgebra morphism 
    then for all $t \in T$, $|\InfPath(t)| = |\InfPath(f(t))|$.
\end{apxpropositionrep}
\begin{proof}
    It can be verified from the definition of analytic functors that if $t_0$ is a successor of $t$, then $f(t_0)$ is a successor of $f(t)$. Moreover, if $s_0$ is a successor of $f(t)$ with multiplicity $k$, and $t_1, \dotsc, t_n$ are all the successors of $t$ that get mapped to $s_0$ by $f$, then $k$ is the sum of the multiplicities of $t_1, \dotsc, t_n$. These properties together imply that there is a bijection between infinite paths from $t$ and infinite paths from $s$. 
\end{proof}

It follows from the last proposition that thin coalgebras are closed under homomorphic images. It is straightforward to see that thin coalgebras are also closed under subcoalgebras and coproducts, therefore they form a covariety~\cite{Rutten:TCS2000}. We will later give a syntactic characterisation of this covariety, via the image of its members into the final $F$-coalgebra.

\begin{example}
\label{ex:morphismPreservationFailure}
    Consider the finitary covariant power set functor $\Pow$ defined as $\Pow(X) = \{ Y \subseteq X \mid \text{$Y$ is finite}\}$ and $\Pow(f)(Y) = \{ f(y) \mid y \in Y \}$. Although this is a non-analytic functor, one could analogously define a path in a $\Pow$-coalgebra (note that all multiplicities would be $1$) and thin $\Pow$-coalgebras. There exists a $\Pow$-coalgebra morphism from the full binary tree (with uncountably-many infinite paths) to a single reflexive point (with only a single infinite path). This shows that \Cref{prop:morphismsPreservePathCount} fails for $\Pow$, hence thin $\Pow$-coalgebras would not form a covariety. This justifies excluding $\Pow$ from the scope of our definitions.
\end{example}

We finish the section with a description of finite thin coalgebras via \emph{cycles}. 

\begin{definition}[Cycle]
    Let $(T,\tau)$ be an $F$-coalgebra and $t \in T$. A \emph{cycle} through $t$ is a finite path from $t$ to $t$. Two cycles are \emph{comparable} if one is a prefix of the other.
\end{definition}

\begin{apxpropositionrep}\label{prop:thin-cycles}
    Let $(T,\tau)$ be a \emph{finite} coalgebra and $t \in T$. Then $t$ is thin if and only if for all $t' \in T$ that are reachable from $t$ by a finite path, all cycles through $t'$ are comparable. This condition can be checked in linear time in the number of nodes and edges in the successor-multigraph of  $(T,\tau)$.
\end{apxpropositionrep}
\begin{proof}
    We first prove the equivalence.
    
    ($\Rightarrow$) We reason by contraposition. Assume there is a finite path $\pi$ from $t$ to some $t'$ and $\pi_1$, $\pi_2$ are two incomparable cycles through $t' \in T$. For any infinite sequence $(n_i)_{i \in \omega} \in \{0, 1\}^\omega$, we can construct an infinite path from $t$ by composing $\pi, \pi_{n_0}, \pi_{n_1}, \pi_{n_2}, \dotsc$. The assumption that $\pi_1$ and $\pi_2$ are incomparable guarantees that for each sequence $(n_i)_{i \in \omega}$ we get a distinct infinite path. Hence we have constructed uncountably many infinite paths from $t$, and so $t$ is not thin.

    ($\Leftarrow$) Assume that for every state $t'$ reachable from $t$, all cycles through $t'$ are comparable. We will show that every infinite path from $t$ is uniquely determined by some finite prefix thereof. Since there are countably many finite prefixes, this will imply that there are countably many infinite paths from $t$.

    Recall from graph theory that two vertices are \emph{strongly connected} if there exists a path from one vertex to the other and vice-versa. Every graph can be partitioned into \emph{strongly connected components}, which are maximal sets of strongly connected vertices. Note that this definition can be readily applied to $F$-coalgebras as well.

    Consider any infinite path $\pi$ from $t$. Since $T$ is finite and the set of strongly connected components is partially ordered, we know that after a certain point $i$, all states in $\pi$ belong to the same strongly connected component $C$. Let $s$ be the $i$-th state in $\pi$ and $\pi_0$ be the shortest cycle through $s$ in $C$. By our assumption, all other cycles through $s$ in $C$ are comparable to $\pi_0$, so they are obtained by composing $\pi_0$ finitely many times. Therefore, there exists only one infinite path from $s$ in $C$, obtained by composition $\pi_0$ infinitely many times. In other words, $\pi$ is uniquely determined by its prefix up to $i$.

    The condition that all cycles through a state are comparable can be verified in linear time in the size of $(T,\tau)$. Here by size of $(T,\tau)$ we mean the number of states plus the number of edges. It suffices to check that for all strongly connected components $C$ reachable from $t$, all cycles in $C$ through the same vertex are comparable. Strongly connected components can be found in linear time (e.g., using Tarjan's algorithm). For each connected component, comparability of cycles is equivalent to the property that there exists a unique shortest path between each two states. The latter property can be checked in linear time with a simple graph traversal.
\end{proof}


\section{Syntax, semantics and constructible behaviours}
\label{sec:alg-thin-rep}

Thin coalgebras can alternatively be characterised via the notion of \emph{constructible behaviours}, the main topic of the current section. Constructible behaviours are elements of the final $F$-coalgebra that can be represented using a certain infinitary \emph{syntax}, arising as the initial algebra of a suitable $\Set$-endofunctor.
We refer to elements of this initial algebra as \emph{terms}.
We equip the syntax with \emph{operational semantics} and \emph{denotational semantics}, both of which interpret terms as elements of the final $F$-coalgebra, and we show in \Cref{prop:interprIsCoalgMorphism} that these semantics coincide. An element of the final $F$-coalgebra is \emph{constructible} if it has a term representative, i.e., if it is the semantics of some term.
After developing the necessary machinery, we show in \Cref{sec:thinCoincidesWithConstructible} that thin coalgebra behaviours are precisely the constructible behaviours (\Cref{thm:thinIFFConstructible}), hence obtaining a syntactic characterisation of thinness.

We start with the definition of the syntax. Each term is assembled from simpler terms using two possible constructors. The first constructor corresponds to the given functor $F$, while the second constructor is an infinite product of $F$-contexts, which will be interpreted as an infinite branch.

\begin{definition}[Syntax]
\label{def:FplusG}
    We define the $\Set$-endofunctor $G$ as the composition $G = (-)^\omega \circ F'$. That is, for a set $X$, $G(X) \coloneqq (F'X)^\omega$ is the set of all streams of contexts over $X$. The \emph{syntax} is given by the initial $(F+G)$-algebra and the \emph{terms} are elements of this algebra.
\end{definition}

\begin{notation}
    Given an $(F+G)$-algebra $(C,\gamma)$, we write $\gamma = [\gamma_0, \gamma_1]$ where $\gamma_0: FC \to C$ and $\gamma_1: GC \to C$. 
    We write $(A,\alpha)$ for the initial $(F+G)$-algebra.
    Since $\alpha$ is bijective, $\Im(\alpha_0)$ and $\Im(\alpha_1)$ partition $A$, so we define the set $\Branch \coloneqq \Im(\alpha_0)$ of \emph{$F$-terms}, and the set $\Stream \coloneqq \Im(\alpha_1)$ of \emph{$G$-terms}. Lastly, let $(Z, \zeta)$ denote the final $F$-coalgebra.
\end{notation}

\begin{remark}
    The initial $(F+G)$-algebra and the final $F$-coalgebra exist by \cite[Theorem~6.10]{AdamekMiliusMoss2018FixedPointsOfFunctors}, because $F$ and $G$ are bounded (hence accessible) functors.
\end{remark}

Note that, technically speaking, our terms are not purely syntactic objects, but rather equivalence classes (orbits) under the permutation group actions given by the analytic functor $F$.

Our aim is to interpret elements of $A$ (the syntactic objects) in $Z$ (the semantic objects), guided by the following intuition. Each $F$-term is constructed from some $\bar a \in FA$, so there is an obvious way to interpret it as a state with transition type $F$. A $G$-term is constructed from a stream of contexts $(\bar a_n')_{n \in \omega} \in (F'A)^\omega = GA$ and its $F$-behaviour is obtained by plugging the interpretation of the stream $(\bar a'_m)_{m > n}$ into the context $\bar a'_n$, for each $n \in \omega$, thus forming an infinite path along the context holes. In other words, we compose sequentially all contexts in the stream to give rise to a new infinite path.

This amounts to defining an \emph{operational semantics} for our syntax. It is defined by coinduction, i.e., by defining an $F$-coalgebra structure $\epsilon: A \to FA$ and using finality of $(Z,\zeta)$. While the definition of $\epsilon$ on $F$-terms is straightforward, in order to define $\epsilon$ on $G$-terms, we use the fact that in every $(F+G)$-algebra $(C,\gamma)$ there is a natural way to move from $GC$ to $FC$ by plugging the evaluation of the tail into the head.

\begin{definition}\label{def:branch}
    For an $(F+G)$-algebra $(C,\gamma = [\gamma_0, \gamma_1])$, the map $\branch_\gamma: G C \to F C$ is defined by:
    \begin{equation}\label{eq:branch}
        \branch_\gamma \coloneqq \trig_C \circ \langle \head, \gamma_1 \circ \tail \rangle.
    \end{equation}    
\end{definition}

\begin{definition}[Operational semantics]
\label{def:FcoalgebraOnA}
    Define an $F$-coalgebra structure map $\epsilon$ on $A$ as follows:
    \begin{align*}
        \epsilon &\coloneqq A \xrightarrow{\alpha^{-1}} FA + GA \xrightarrow{[\id_{FA}, \branch_\alpha]} FA.
    \end{align*}
    We write $\interprop-$ for the unique $F$-coalgebra morphism from $(A, \epsilon)$ to $(Z, \zeta)$ and call it the \emph{operational semantics}.
\end{definition}

\begin{definition}[Constructible behaviour]
    A term $a \in A$ is a \emph{representative} of a behaviour $z \in Z$ if $\interprop a = z$. The set $\thin Z \subseteq Z$ of \emph{constructible behaviours} consists of all elements of $Z$ that have a representative.
\end{definition}

\begin{example}
\label{ex:constructibleBehaviourAndRepresentatives}
Consider again $F(X) \coloneqq X + X^2$ from \Cref{ex:polynomialFunctor,ex:polynomialFunctorDerivative}. \Cref{fig:representatives} depicts a constructible behaviour (drawn as an $F$-tree) and two of its representatives. Note that the trees in \Cref{fig:someRepresentative,fig:normalRepresentative} are not the syntax trees of the representatives.
They are annotated versions of the tree in \Cref{fig:thinTree} where the annotations indicate how the tree is obtained from two different representatives. A node annotated with $G$ is the root of a sub-tree that is represented as an infinite product of contexts (with context holes indicated using squares and non-hole successors indicated with dashed lines). As described prior to \Cref{def:branch}, an infinite product of contexts is interpreted as an $F$-tree by composing the contexts via the plug-in (i.e., context $n+1$ is plugged into the hole of context $n$). A node annotated with $F$ is the root of a sub-tree that is represented using an $F$-operation at the top level.
The labels $op_1$ and $op_2$ have been omitted, since they can be inferred from the branching degree. 

\end{example}

\begin{figure}
    \centering
    \begin{subfigure}[b]{0.5\textwidth}
        \centering
        \scalebox{0.5}{
        \hspace{-50px}
        \begin{tikzpicture}
            \draw (5, 7) -- (2.5, 6);
            \draw (5, 7) -- (7.5, 6);
            \draw (2.5, 6) -- (1.5, 5);
            \draw (2.5, 6) -- (3.5, 5);
            \draw (1.5, 5) -- (1., 4.5);
            \draw (3.5, 5) -- (2.5, 4);
            \draw (2.5, 4) -- (2., 3.5);
            \draw (3.5, 5) -- (4.5, 4);
            \draw (4.5, 4) -- (3.5, 3.);
            \draw (3.5, 3) -- (3.0, 2.5);
            \draw (4.5, 4) -- (5., 3.5);
            \draw (7.5, 6) -- (6.5, 5);
            \draw (7.5, 6) -- (8.5, 5);
            \draw (6.5, 5) -- (6.5, 4);
            \draw (6.5, 4) -- (6.5, 3.5);
            \draw (6.5, 4) -- (6.5, 3);
            \draw (6.5, 3) -- (6.5, 2.5);
            \draw (8.5, 5) -- (8.5, 4);
            \draw (8.5, 4) -- (8.5, 3.5);
            \draw (8.5, 4) -- (8.5, 3);
            \draw (8.5, 3) -- (8.5, 2.5);
            \filldraw (5, 7) circle (2pt);
            \filldraw (2.5, 6) circle (2pt);
            \filldraw (1.5, 5) circle (2pt);
            \filldraw (3.5, 5) circle (2pt);
            \filldraw (4.5, 4) circle (2pt);
            \filldraw (3.5, 3) circle (2pt);
            \filldraw (2.5, 4) circle (2pt);
            \filldraw (7.5, 6) circle (2pt);
            \filldraw (6.5, 5) circle (2pt);
            \filldraw (6.5, 4) circle (2pt);
            \filldraw (6.5, 3) circle (2pt);
            \filldraw (8.5, 5) circle (2pt);
            \filldraw (8.5, 4) circle (2pt);
            \filldraw (8.5, 3) circle (2pt);
            \node at (0.7, 4.2) {$ \rotatebox[origin=c]{-45}{\vdots}$};
            \node at (1.7, 3.2) {$ \rotatebox[origin=c]{-45}{\vdots}$};
            \node at (2.7, 2.2) {$ \rotatebox[origin=c]{-45}{\vdots}$};
            \node at (4.7, 2.8) {$ \rotatebox[origin=c]{45}{\vdots}$};
            \node at (5.3, 3.2) {$ \rotatebox[origin=c]{45}{\vdots}$};
            \node at (6.5, 2.2) {$ \rotatebox[origin=c]{0}{\vdots}$};
            \node at (8.5, 2.2) {$ \rotatebox[origin=c]{0}{\vdots}$};
        \end{tikzpicture}
        }
        \caption{$F$-tree}
        \label{fig:thinTree}
    \end{subfigure}
    \begin{subfigure}[b]{0.24\textwidth}
        \centering
        \scalebox{0.5}{
        \begin{tikzpicture}
            \draw[branchLine] (5, 7) -- (2.5, 6);
            \draw[branchLine] (5, 7) -- (7.5, 6);
            \draw[dashed] (2.5, 6) -- (1.5, 5);
            \draw[streamLine] (2.5, 6) -- (3.5, 5);
            \draw[streamLine] (1.5, 5) -- (1., 4.5);
            \draw[dashed] (3.5, 5) -- (2.5, 4);
            \draw[streamLine] (2.5, 4) -- (2., 3.5);
            \draw[streamLine] (3.5, 5) -- (4.5, 4);
            \draw[dashed] (4.5, 4) -- (3.5, 3.);
            \draw[streamLine] (3.5, 3) -- (3, 2.5);
            \draw[streamLine] (4.5, 4) -- (5., 3.5);
            \draw[branchLine] (7.5, 6) -- (6.5, 5);
            \draw[branchLine] (7.5, 6) -- (8.5, 5);
            \draw[branchLine] (6.5, 5) -- (6.5, 4);
            \draw[streamLine] (6.5, 4) -- (6.5, 3.5);
            \draw[streamLine] (6.5, 4) -- (6.5, 3);
            \draw[streamLine] (6.5, 3) -- (6.5, 2.5);
            \draw[streamLine] (8.5, 5) -- (8.5, 4);
            \draw[streamLine] (8.5, 4) -- (8.5, 3.5);
            \draw[streamLine] (8.5, 4) -- (8.5, 3);
            \draw[streamLine] (8.5, 3) -- (8.5, 2.5);
            \node[branch] at (5, 7) {$F$};
            \node[stream] at (2.5, 6) {$G$};
            \node[stream] at (1.5, 5) {$G$};
            \filldraw[color=white] (3.5-0.1,5-0.1) rectangle ++(0.2,0.2);
            \draw[draw=blue] (3.5-0.1,5-0.1) rectangle ++(0.2,0.2);
            \filldraw[color=white] (4.5-0.1,4-0.1) rectangle ++(0.2,0.2);
            \draw[draw=blue] (4.5-0.1,4-0.1) rectangle ++(0.2,0.2);
            \node[stream] at (2.5, 4) {$G$};
            \node[stream] at (3.5, 3) {$G$};
            \node[branch] at (7.5, 6) {$F$};
            \node[branch] at (6.5, 5) {$F$};
            \node[stream] at (6.5, 4) {$G$};
            \filldraw[color=white] (6.5-0.1,3-0.1) rectangle ++(0.2,0.2);
            \draw[draw=blue] (6.5-0.1,3-0.1) rectangle ++(0.2,0.2);
            \node[stream] at (8.5, 5) {$G$};
            \filldraw[color=white] (8.5-0.1,4-0.1) rectangle ++(0.2,0.2);
            \draw[draw=blue] (8.5-0.1,4-0.1) rectangle ++(0.2,0.2);
            \filldraw[color=white] (8.5-0.1,3-0.1) rectangle ++(0.2,0.2);
            \draw[draw=blue] (8.5-0.1,3-0.1) rectangle ++(0.2,0.2);
            \node at (0.7, 4.2) {$ \rotatebox[origin=c]{-45}{\vdots}$};
            \node at (1.7, 3.2) {$ \rotatebox[origin=c]{-45}{\vdots}$};
            \node at (2.7, 2.2) {$ \rotatebox[origin=c]{-45}{\vdots}$};
            \node at (4.7, 2.8) {$ \rotatebox[origin=c]{45}{\vdots}$};
            \node at (5.3, 3.2) {$ \rotatebox[origin=c]{45}{\vdots}$};
            \node at (6.5, 2.2) {$ \rotatebox[origin=c]{0}{\vdots}$};
            \node at (8.5, 2.2) {$ \rotatebox[origin=c]{0}{\vdots}$};
        \end{tikzpicture}
        }
        \caption{Representative 1}
        \label{fig:someRepresentative}
    \end{subfigure}
    \begin{subfigure}[b]{0.24\textwidth}
        \centering
        \scalebox{0.5}{
        \hspace{4px}
        \begin{tikzpicture}
            \draw[streamLine] (5, 7) -- (2.5, 6);
            \draw[dashed] (5, 7) -- (7.5, 6);
            \draw[dashed] (2.5, 6) -- (1.5, 5);
            \draw[streamLine] (2.5, 6) -- (3.5, 5);
            \draw[streamLine] (1.5, 5) -- (1., 4.5);
            \draw[dashed] (3.5, 5) -- (2.5, 4);
            \draw[streamLine] (2.5, 4) -- (2., 3.5);
            \draw[streamLine] (3.5, 5) -- (4.5, 4);
            \draw[dashed] (4.5, 4) -- (3.5, 3.);
            \draw[streamLine] (3.5, 3) -- (3, 2.5);
            \draw[streamLine] (4.5, 4) -- (5., 3.5);
            \draw[branchLine] (7.5, 6) -- (6.5, 5);
            \draw[branchLine] (7.5, 6) -- (8.5, 5);
            \draw[streamLine] (6.5, 5) -- (6.5, 4);
            \draw[streamLine] (6.5, 4) -- (6.5, 3.5);
            \draw[streamLine] (6.5, 4) -- (6.5, 3);
            \draw[streamLine] (6.5, 3) -- (6.5, 2.5);
            \draw[streamLine] (8.5, 5) -- (8.5, 4);
            \draw[streamLine] (8.5, 4) -- (8.5, 3.5);
            \draw[streamLine] (8.5, 4) -- (8.5, 3);
            \draw[streamLine] (8.5, 3) -- (8.5, 2.5);
            \node[stream] at (5, 7) {$G$};
            \filldraw[color=white] (2.5-0.1,6-0.1) rectangle ++(0.2,0.2);
            \draw[draw=blue] (2.5-0.1,6-0.1) rectangle ++(0.2,0.2);
            \node[stream] at (1.5, 5) {$G$};
            \filldraw[color=white] (3.5-0.1,5-0.1) rectangle ++(0.2,0.2);
            \draw[draw=blue] (3.5-0.1,5-0.1) rectangle ++(0.2,0.2);
            \filldraw[color=white] (4.5-0.1,4-0.1) rectangle ++(0.2,0.2);
            \draw[draw=blue] (4.5-0.1,4-0.1) rectangle ++(0.2,0.2);
            \node[stream] at (2.5, 4) {$G$};
            \node[stream] at (3.5, 3) {$G$};
            \node[branch] at (7.5, 6) {$F$};
            \node[stream] at (6.5, 5) {$G$};
            \filldraw[color=white] (6.5-0.1,4-0.1) rectangle ++(0.2,0.2);
            \draw[draw=blue] (6.5-0.1,4-0.1) rectangle ++(0.2,0.2);
            \filldraw[color=white] (6.5-0.1,3-0.1) rectangle ++(0.2,0.2);
            \draw[draw=blue] (6.5-0.1,3-0.1) rectangle ++(0.2,0.2);
            \node[stream] at (8.5, 5) {$G$};
            \filldraw[color=white] (8.5-0.1,4-0.1) rectangle ++(0.2,0.2);
            \draw[draw=blue] (8.5-0.1,4-0.1) rectangle ++(0.2,0.2);
            \filldraw[color=white] (8.5-0.1,3-0.1) rectangle ++(0.2,0.2);
            \draw[draw=blue] (8.5-0.1,3-0.1) rectangle ++(0.2,0.2);
            \node at (0.7, 4.2) {$ \rotatebox[origin=c]{-45}{\vdots}$};
            \node at (1.7, 3.2) {$ \rotatebox[origin=c]{-45}{\vdots}$};
            \node at (2.7, 2.2) {$ \rotatebox[origin=c]{-45}{\vdots}$};
            \node at (4.7, 2.8) {$ \rotatebox[origin=c]{45}{\vdots}$};
            \node at (5.3, 3.2) {$ \rotatebox[origin=c]{45}{\vdots}$};
            \node at (6.5, 2.2) {$ \rotatebox[origin=c]{0}{\vdots}$};
            \node at (8.5, 2.2) {$ \rotatebox[origin=c]{0}{\vdots}$};
        \end{tikzpicture}
        }
        \caption{Representative 2}
        \label{fig:normalRepresentative}
    \end{subfigure}
    \caption{An $F$-tree and two annotated versions indicating different representatives thereof, cf.~\Cref{ex:constructibleBehaviourAndRepresentatives}.}
    \label{fig:representatives}
\end{figure}

Additionally, one can define a \emph{denotational semantics} for the syntax by induction, i.e., by equipping $Z$ with an $(F+G)$-algebra structure $\beta = [\beta_0,\beta_1]: FZ + GZ \to Z$ and using initiality of $(A, \alpha)$. While $\beta_0 : F Z \to Z$ is immediately defined as $\zeta^{-1}$, the definition of $\beta_1$ is more involved. Specifically, the map $\beta_1 : G Z \to Z$ is obtained by coinduction, i.e., by endowing the set $Z + G Z$ with $F$-coalgebra structure, and then exploiting the finality of $(Z,\zeta)$. (The more direct attempt to endow $G Z$ by itself with an $F$-coalgebra structure fails.)

\begin{definition}[Denotational semantics]
\label{def:F+GalgebraOnZ}
Consider the $F$-coalgebra $(Z + GZ,[\xi_0,\xi_1])$ with $\xi_0 = F\inj[Z+GZ]1 \circ \zeta$ and $\xi_1$ given by:
 \begin{multline*}
     GZ \xrightarrow{\langle \head,\tail\rangle} F'Z \times GZ \xrightarrow{F' \inj[Z+GZ]1 \times \inj[Z+GZ]2} \\
     F'(Z+GZ) \times (Z+GZ) \xrightarrow{\triangleright_{Z + G Z}} F(Z+GZ).
 \end{multline*}
Let $[\id_Z,\beta_1] : (Z + G Z,[\xi_0,\xi_1]) \to (Z,\zeta)$ be the unique $F$-coalgebra morphism arising from the finality of $(Z,\zeta)$.
(Uniqueness of morphisms into $(Z,\zeta)$ forces the first component of the above morphism to be $\id_Z$.)
 \[
 \UseComputerModernTips\xymatrix{
 Z + G Z \ar[r]^-{[\xi_0,\xi_1]} \ar@{-->}[d]_-{[\id_Z,\beta_1]} & F(Z + G Z) \ar@{-->}[d]^-{F[\id_Z,\beta_1]}\\
 Z \ar[r]_-{\zeta} & F Z
 }
 \]
Now define an $(F+G)$-algebra structure $\beta = [\beta_0,\beta_1]$ on $Z$ by taking $\beta_0 := \zeta^{-1} : F Z \to Z$ and $\beta_1 : G Z \to Z$ as above. We write $\interprden-$ for the unique $(F+G)$-algebra morphism from $(A, \alpha)$ to $(Z, \beta)$ and call it the \emph{denotational semantics}. 
\end{definition}

Given an element of $GZ$ (i.e., a stream of contexts), the map $\xi_1$ extracts the required $F$-structure by plugging the tail of the stream, viewed as an element of $Z + G Z$, into the head of the stream, viewed as a context over $Z + G Z$. The move to $Z + G Z$ is needed in order to apply the plug-in operation.

Next, we prove that the operational and denotational semantics coincide by showing that $\interprden-: (A, \alpha) \to (Z,\beta)$ is also an $F$-coalgebra morphism $\interprden-: (A,\epsilon) \to (Z,\zeta)$.
The proof relies on the $(F+G)$-algebra structure on $Z$ being \emph{coherent}, in the following sense:
\begin{definition}[Coherence]
\label{def:coherence}
    An $(F+G)$-algebra $(C,\gamma=[\gamma_0,\gamma_1])$ is \emph{coherent} if it satisfies the equation:
    \begin{equation}\label{eq:plugging}
    \gamma_1 = \gamma_0 \circ \branch_\gamma. \tag{$\dagger$}
    \end{equation}
    That is, the following diagram commutes:
    \begin{equation*}   
        \UseComputerModernTips\xymatrix{
         G C \ar[rr]^-{\gamma_1} \ar[d]_-{\langle \head,\gamma_1 \circ \tail \rangle}& & C\\
         F'C \times C \ar[rr]_-{\triangleright_C} & & F C \ar[u]_-{\gamma_0}
         }
    \end{equation*}
\end{definition}
Thus, for each infinite stream of contexts $(\bar c'_n)_{n \in \omega} \in G C$, the equation identifies the infinite product of these contexts $\gamma_1((\bar c'_n)_{n \in \omega})$ with the result of plugging the infinite product of the tail of the stream $\gamma_1((\bar c'_n)_{n > 0})$ into the head of the stream $\bar c'_0$.
This is similar to the $\omega$-semigroup axiom: $c_0 \cdot \Pi(c_1,c_2, \ldots) = \Pi(c_0, c_1,c_2, \ldots)$~\cite{PerrinPin:2004:InfiniteWords}.

\begin{apxlemmarep}\label{lem:Z-beta-coh}
$(Z,\beta)$ is coherent.
\end{apxlemmarep}
\begin{proof}
    Using the fact that $\beta_0 = \zeta^{-1} : F Z \to Z$ is an isomorphism, coherence of $(Z,\beta)$ amounts to commutativity of the outer part of the following diagram:
    \[\begin{tikzcd}
        GZ & {F' Z \times G Z} & {F' Z \times Z} \\
        & {F'(Z + G Z) \times (Z + G Z)} \\
        & {\text{(a)}\hspace{28px}F(Z+GZ)\hspace{28px}} \\
        Z && FZ
        \arrow["{\langle \head,\tail\rangle}", from=1-1, to=1-2]
        \arrow["{\beta_1}"', from=1-1, to=4-1]
        \arrow["{\id_{F' Z} \times \beta_1}", from=1-2, to=1-3]
        \arrow["{F' \inj[Z+GZ]1 \times \inj[Z+GZ]2}"', from=1-2, to=2-2]
        \arrow["{\trig_Z}", from=1-3, to=4-3]
        \arrow["\fcolorbox{white}{white}{$F'[\id_Z,\beta_1] \times [\id_Z,\beta_1]$}"', from=2-2, to=1-3]
        \arrow["{\trig_{Z+GZ}}", from=2-2, to=3-2]
        \arrow["{F[\id_Z, \beta_1]}", from=3-2, to=4-3]
        \arrow["\zeta"', from=4-1, to=4-3]
    \end{tikzcd}\]
    Here, the (a) part commutes by the definition of $\beta_1$, the triangle commutes trivially, and the trapezium commutes by the naturality of $\triangleright$. Consequently, the outer part commutes. 
\end{proof}

\begin{apxlemmarep}\label{lem:branchFG}
    Let $(C, \gamma)$ and $(D, \delta)$ be $(F+G)$-algebras and let $f:(C, \gamma) \to (D, \delta)$
    be an algebra morphism. We have $\branch_\delta \circ Gf = Ff \circ \branch_\gamma$.
\end{apxlemmarep}
\begin{proof}
    Consider the following diagram 
    \[
    \xymatrix@R=1em@C=1.5em{GC \ar[rr]^-{\langle \head,\tail \rangle} \ar[dd]_{Gf} & & F'C \times GC 
    \ar[rr]^{\id \times \gamma_1} \ar[dd]^{F' f \times Gf} & & F'C \times C \ar[rr]^{\trig_C} \ar[dd]^{F' f \times f} & & FC 
    \ar[dd]_{F f}
    \\
    \\
    GD \ar[rr]_-{\langle \head,\tail \rangle} & & F'D \times GD \ar[rr]_{\id \times \delta_1} & & F'D \times D \ar[rr]_{\trig_D}  & & FD }
    \]
    The left square commutes by the definition of $Gf$, the right square commutes by naturality 
    of $\trig$ and the middle square commutes because $f$ is an $(F+G)$-algebra morphism. The outer square implies immediately the claim of the lemma as the upper horizontal arrows compose to 
    $\branch_{\gamma}$ while the lower ones compose to $\branch_{\delta}$. 
\end{proof}

We can now prove that $\interprden-$ preserves the $F$-coalgebra structure.

\begin{proposition}
\label{prop:interprIsCoalgMorphism}
$\interprden - : (A, \epsilon) \to (Z, \zeta)$ is an $F$-coalgebra morphism. Therefore, $\interprden- = \interprop-$.
\end{proposition}
\begin{proof}
The statement follows from the commutativity of the following diagram:
 \begin{equation}
 \label{eqn:iterprIsCoalgMorphism}
 \UseComputerModernTips\xymatrix{
 A \ar[d]^-{\interprden -} \ar[rr]^-{\alpha^{-1}} & & F A + G A \ar[d]^-{(F+G) \interprden - } \ar[rr]^-{[\id_{FA},\branch_\alpha]} & & F A \ar[d]^-{F \interprden -} \\
 Z \ar@/^{-1pc}/[rrrr]_-{\zeta} & & F Z + G Z \ar[ll]_-{[\beta_0,\beta_1]} \ar[rr]^-{[\id_{F Z},\branch_\zeta]} & & F Z
 }
 \end{equation}
 The left square commutes since $\interprden -$ is an $(F+G)$-algebra morphism. Commutativity of the right square follows from \Cref{lem:branchFG}.
 Finally, commutativity of the lower crescent in (\ref{eqn:iterprIsCoalgMorphism}) follows from $\beta_0 = \zeta^{-1}$ and the fact that $(Z,\beta)$ is coherent.
\end{proof}

Thus from now on we simply write $\interpr-$ for $\interprden-$ or $\interprop-$ and call it the \emph{semantics map}. Its properties are summarised by the following diagram, where both inner rectangles commute.
\[\begin{tikzcd}
	{(F+G)A} && {(F+G)Z} \\
	A && Z \\
	FA && FZ
	\arrow["{(F+G)\interpr-}", from=1-1, to=1-3]
	\arrow["\alpha"', from=1-1, to=2-1]
	\arrow["\beta", from=1-3, to=2-3]
	\arrow["{\interpr-}", from=2-1, to=2-3]
	\arrow["\epsilon"', from=2-1, to=3-1]
	\arrow["\zeta", from=2-3, to=3-3]
	\arrow["{F\interpr-}"', from=3-1, to=3-3]
\end{tikzcd}\]

The fact that $\interpr-$ is an $(F+G)$-algebra morphism implies that $\thin Z = \Im(\interpr-)$ is an $(F+G)$-subalgebra of $(Z,\beta)$.

\begin{definition}
     Let $(\thin Z, \thin \beta)$ be the $(F+G)$-subalgebra of $(Z,\beta)$ determined by the image of $A$ under $\interpr-$.
\end{definition}

In \Cref{sec:main-results} we will see that \Cref{eq:plugging} serves as an axiomatisation of $(\thin Z,\thin \beta)$, i.e., two terms have the same semantics precisely when they are identified by the equation.


\section{Normal terms}
\label{sec:normal}

A constructible behaviour generally has many representatives, i.e,
for an element $z \in \thin Z$, there are many $a \in A$ with $\interpr a = z$. In this section, we present a way to choose a canonical representative for each $z \in \thin Z$. We call these canonical terms \emph{normal}. We show that every constructible element has a unique normal representative. This statement is split into two propositions, Existence of normal representatives (\Cref{prop:existenceOfNormalTerms}) and Uniqueness of normal representatives (\Cref{prop:uniquenessOfNormalTerms}). Thus the collection of normal terms is in one-to-one correspondence with constructible elements of $Z$.

The definition of normality is based on a notion of \emph{rank} of a term, which measures the nesting depth of the term. The rank is a pair consisting of a \emph{major rank} and a \emph{minor rank}. The major rank is the maximal number of nested $G$-constructors, while the minor rank is the maximal number of nested $F$-constructors before reaching the first $G$-constructor. In preparation, we first define what a \emph{subterm} is.

\begin{definition}[Subterms]
    For $a \in A$, we define the set $\subterm(a) \subseteq A$ of \emph{subterms of $a$} as
    $\subterm(a) \coloneq \Base_{F+G}(\alpha^{-1}(a))$.
\end{definition}

Note that the ``subterm of'' relation on $A$ is well-founded, because a subterm $a$ of $b$ appears earlier than $b$ in the initial sequence construction of the initial algebra $(A, \alpha)$.

\begin{definition}
\label{def:majorAndMinorRank}
    The \emph{major rank} of $a \in A$ is an ordinal $\majrk(a)$ defined as:
    \begin{equation*}
        \majrk(a) \coloneqq \begin{dcases}
            \sup\{\majrk(b)\mid b \in \subterm(a) \} + 1 & \text{$a \in \Stream$} \\
            \sup\{\majrk(b)\mid b \in \subterm(a) \}& \text{$a \in \Branch$}.
        \end{dcases}
    \end{equation*}
    The \emph{minor rank} of $a$ is an ordinal $\minrk(a)$ defined as:
    \begin{equation*}
        \minrk(a) \coloneqq \begin{dcases}
            0 & a \in \Stream \\
            \sup\{ \minrk(b) \mid {b \in \subterm(a)} \} + 1 &  a \in \Branch.
        \end{dcases}
    \end{equation*}
    The \emph{rank} of $a$ is the pair of ordinals $\rk(a) \coloneqq (\majrk(a), \minrk(a))$.
\end{definition}

\begin{notation} 
We write $\preceq$ for the lexicographic ordering on pairs of ordinals,
and $\prec$ for the corresponding strict order. 
We use the following notation: $\Majrk_i \coloneqq \{ a \in A \mid \majrk(a) = i \}$ and $\Rk_{(i,j)} \coloneqq \{ a \in A \mid \rk(a) = (i, j) \}$. Analogously, we write $\Majrk_{\leq i}$, $\Majrk_{<i}$, $\Rk_{\preceq(i,j)}$ and $\Rk_{\prec(i,j)}$ with the obvious denotations.
\end{notation}

\begin{example}
    Consider again the representatives in \Cref{fig:representatives}. The term illustrated in \Cref{fig:someRepresentative} has rank $(2, 3)$, and the term illustrated in \Cref{fig:normalRepresentative} has rank $(2, 0)$. 
\end{example}

\begin{observation}
\label{obs:majorRankInclusions}
    We have the following useful relations between ranks:
    \begin{align*}
        &\forall \bar{a} \in FA: \alpha_0(\bar a) \in \Rk_{(i,j)} 
        \implies \bar a \in F(\Rk_{\prec(i,j)}),  
        \\
        &\forall \bar{a} \in GA: \alpha_1(\bar a) \in \Majrk_i 
        \;\,\implies \bar a \in G(\Majrk_{< i}),
        \\
        &\forall a \in A: a \in \Rk_{(i,j)} \implies \subterm(a) \subseteq \Rk_{\prec(i,j)}, \\
        &\alpha_0[F(\Majrk_{\leq i})] \subseteq \Majrk_{\leq i},
        \hspace{25px} \alpha_1[G(\Majrk_{<i})] \subseteq \Majrk_{\leq i}.
    \end{align*}
\end{observation}

The normal representative of a given constructible behaviour is the lowest-ranked representative that is composed of normal subterms.

\begin{definition}[Normal term]
\label{def:normalTerm}
    Define the set of normal terms $\Normal \subseteq A$ by induction on the rank. A term $a \in \Rk_{(i,j)}$ is said to be normal if:
    \begin{enumerate}[(i)]
        \item \label{item:normalitySubterms} every element of $\subterm(a) \subseteq \Rk_{\prec(i,j)}$ is normal, and
        \item \label{item:normalityRank} there is no $b \in \Rk_{\prec(i,j)} \cap \Normal$ with $\interpr b = \interpr a$.
    \end{enumerate}
\end{definition}

\begin{example}
    The term illustrated in \Cref{fig:normalRepresentative} is normal, while the term illustrated in  \Cref{fig:someRepresentative} is not.
\end{example}

We proceed with the key properties of normal terms: Existence and Uniqueness of normal representatives.

\begin{toappendix}
\paragraph*{Existence of Normal Representatives}
    Let $f: A \to A$ be a map, $g: (A,\alpha) \to (C,\gamma)$ be an algebra morphism and $a \in A$. The lemma below shows that if $g(b) = (g \circ f)(b)$ for all $b \in \subterm(a)$, then $g(a) = (g \circ f)(a)$. For example, if $g = \interpr-$, then the statement becomes: if $f$ preserves the semantics of the subterms of $a$, then it preserves the semantics of $a$.

    \begin{lemma}
    \label{lem:subtermPreservationImpliesTermPreservation}
        Let $f: A \to A$, $B \subseteq A$. Let $(C,\gamma)$ be an $(F+G)$-algebra and $g:(A,\alpha) \to (C,\gamma)$ be an $(F+G)$-algebra morphism such that $g(b) = (g\circ f)(b)$ for every $b \in B$. Then for every $\bar b \in (F+G)B$ we have $(g\circ \alpha)(\bar b) = (g \circ (F+G)f) (\bar b)$.
    \end{lemma}
    \begin{proof}
         Without loss of generality, $B = A$, for otherwise we take:
        \begin{equation*}
                f'(a) \coloneqq \begin{cases}
                    f(a) & \text{if $a \in B$} \\
                    a & \text{otherwise}.
                \end{cases}
            \end{equation*}
        instead of $f$ and use the fact that, since $F+G$ preserves inclusions, $(F+G)f(\bar b) = (F+G)f'(\bar b)$ for any $\bar b \in (F+G)B$.
    
        In the following diagram the triangles commute by the assumption on $f$ and functoriality of $F+G$. The trapezia commute since $g$ is an $(F+G)$-algebra morphism. Hence all paths from $(F+G)A$ (top-left) to $C$ commute, from which the lemma follows. (Note that the outer rectangle might not commute.) \qedhere
        \[\begin{tikzcd} %
            {(F+G)A} && {(F+G)A} \\
            & {(F+G)C} \\
            & C \\
            A && A
            \arrow["{(F+G)g}"'{pos=0.6}, from=1-1, to=2-2]
            \arrow["{(F+G)g}"{pos=0.6}, from=1-3, to=2-2]
            \arrow["{(F+G)f}", from=1-1, to=1-3]
            \arrow["g", from=4-1, to=3-2]
            \arrow["g"', from=4-3, to=3-2]
            \arrow["{f}"', from=4-1, to=4-3]
            \arrow["\gamma"', from=2-2, to=3-2]
            \arrow["\alpha", from=1-3, to=4-3]
            \arrow["\alpha"', from=1-1, to=4-1]
        \end{tikzcd} \]
    \end{proof}
\end{toappendix}

\begin{apxpropositionrep}[Existence of normal representatives]
\label{prop:existenceOfNormalTerms}
    There exists a map $\normal: A \to \Normal$ such that for all $a \in A$, $\interpr a = \interpr{\normal(a)}$. Moreover, for all $a \in A$, $\majrk(\normal(a)) \leq \majrk(a)$.
\end{apxpropositionrep}
\begin{proofsketch}
    Define maps $\normal_{(i,j)}: \Rk_{(i,j)} \to \Majrk_{\leq i} \cap \Normal$ with $\interpr{\normal_{(i,j)}(a)} = \interpr a$ by induction on $(i, j)$. Let $a \in \Rk_{(i,j)}$ with $a = \alpha(\bar a)$ for some $\bar a \in (F+G)\Rk_{\prec(i,j)}$. By the induction hypothesis, there is a map $\normal_{\prec(i,j)}\colon \Rk_{\prec(i,j)} \to \Majrk_{\leq i} \cap \Normal$. Let $b \coloneqq (\alpha \circ (F+G)\normal_{\prec(i,j)})(\bar a)$. One can show that $b$ satisfies \Cref{item:normalitySubterms} of normality, $\interpr b = \interpr a$, and $b \in \Majrk_{\leq i}$. Choose $n_{(i,j)}(a)$ to be any element of least rank with these properties. 
\end{proofsketch}
\begin{proof}
    For each pair of ordinals $(i, j)$, we define by induction on $(i, j)$ a map $\normal_{(i,j)}: \Rk_{(i,j)} \to \Majrk_{\leq i} \cap \Normal$ with the property $\forall a \in A : \interpr a = \interpr{\normal_{(i,j)}(a)}$. Suppose we have defined $\normal_{(i',j')}$ for every $(i',j') < (i,j)$. Let $a \in \Rk_{(i,j)}$ with $a = \alpha(\bar a)$ for some $\bar a \in (F+G)\Rk_{\prec(i,j)}$. Define the function:
    \begin{align*}
        \normal_{\prec(i,j)}&: \Rk_{\prec(i,j)} \to \Majrk_{\leq i} \cap \Normal, \\
        \normal_{\prec(i,j)}(c) &\coloneqq \normal_{(i',j')}(c) \quad \text{where} \quad (i',j') \coloneqq \rk(c).
    \end{align*}
    Take $b \coloneqq (\alpha \circ (F+G)\normal_{\prec(i,j)})(\bar a)$. Then $b \in \alpha[(F+G)\Normal]$, so $b$ satisfies \Cref{item:normalitySubterms} of normality. In order to show that $\interpr b = \interpr a$, we define the function:
    \begin{align*}
        h&: A \to A \\
        h(c) &\coloneqq \begin{cases}
            \normal_{\prec(i,j)}(c) & \text{if $c \in \Rk_{\prec(i,j)}$}, \\
            c & \text{otherwise}.
        \end{cases}
    \end{align*}
    By applying \Cref{lem:subtermPreservationImpliesTermPreservation} with $f \coloneqq h, B \coloneqq A, g \coloneqq \interpr-$, we get $\interpr- \circ h \circ \alpha = \interpr- \circ \alpha \circ (F+G)h$. Since $\bar a \in (F+G)\Rk_{\prec(i,j)}$ and $(F+G)$ preserves inclusions, we know $(F+G)h(\bar a) = (F+G)\normal_{\prec(i,j)}(\bar a)$. Therefore:
    \begin{multline*}
        \interpr b = \interpr{(\alpha \circ (F+G)\normal_{\prec(i,j)})(\bar a)} = \\
        \interpr{(\alpha \circ (F+G)h)(\bar a)} = \interpr{(h \circ \alpha)(\bar a)} = \interpr{h(a)} = \interpr{a}.
    \end{multline*}
    Next, we show that $b \in \Majrk_{\leq i}$. We consider two cases.
    \begin{itemize}
        \item $a \in \Branch$. Then $\bar a \in F\Rk_{\prec(i,j)}$ and $b = (\alpha_0 \circ F\normal_{\prec(i,j)})(\bar a)$. Since $F$ preserves inclusions, we know $F\normal_{\prec(i,j)}[F\Rk_{\prec(i,j)}] \subseteq F\Majrk_{\leq i}$, while \Cref{obs:majorRankInclusions} gives us $\alpha_0[F\Majrk_{\leq i}] \subseteq \Majrk_{\leq i}$. Therefore $b \in \Majrk_{\leq i}$.
        \item $a \in \Stream$. Then $\bar a \in G\Majrk_{<i}$ and $b = (\alpha_1 \circ Gn_{\prec(i,j)})(\bar a)$. Now $Gn_{\prec(i,j)}$ $[G\Majrk_{<i}] \subseteq G\Majrk_{<i}$ and $\alpha_1[G\Majrk_{<i}] \subseteq \Majrk_{\leq i}$. Therefore $b \in \Majrk_{\leq i}$.
    \end{itemize}
    
    We have proven that there exists an element $b \in \Majrk_{\leq i}$ with $\interpr b = \interpr a$ satisfying \Cref{item:normalitySubterms} of normality, and we define $\normal_{(i,j)}(a)$ to be any such $b$ of least rank. 
\end{proof}

\begin{toappendix}
    \paragraph*{Uniqueness of Normal Representatives}

    The following lemma states that if a stream of contexts evaluates to a normal term, then its tail also evaluates to a normal term with the same major rank.

    \begin{lemma}
    \label{lem:tailIsNormalWithSameMajRank}
        If $\alpha_1(\bar a) \in \Normal$ for $\bar a \in GA$, then $\majrk((\alpha_1 \circ \tail )(\bar a)) = \majrk(\alpha_1(\bar a))$ and $(\alpha_1 \circ \tail) (\bar a) \in \Normal$.
    \end{lemma}
    \begin{proof}
        Assume towards a contradiction that there exists $b \in A$ with $\interpr b = \interpr{(\alpha_1 \circ \tail)(\bar a))}$ and $\majrk(b) < \majrk(\alpha_1(\bar a))$. We have:
        \begin{align*}
            \hspace{-15px} \interpr{\alpha_1(\bar a)} &= \interpr{(\alpha_0 \circ \trig_A)(\head (\bar a), (\alpha_1 \circ \tail) (\bar a))} \\
            &= (\beta \circ F\interpr- \circ \trig_A)(\head (\bar a), (\alpha_1 \circ \tail) (\bar a)) \\
            &= (\beta \circ \trig_Z)((F'\interpr- \circ \head)(\bar a), \interpr{ (\alpha_1 \circ \tail) (\bar a)}) \\
            &= (\beta \circ \trig_Z)((F'\interpr- \circ \head)(\bar a), \interpr b) \\
            &= (\beta \circ F\interpr- \circ \trig_A)(\head(\bar a), b)  \\
            &= \interpr{(\alpha_0 \circ \trig_A)(\head(\bar a), b)},
        \end{align*}
        where the first equality uses the definition of $\epsilon$ and the fact that $\interpr-$ is an $F$-coalgebra morphism, the second and sixth use the fact that $\interpr-$ is an $(F+G)$-algebra morphism, and the third and fifth use naturality of $\trig$.
        Moreover:
        \begin{align*}
            &\hspace{-15px}\majrk((\alpha_0 \circ \trig_A)(\head (\bar a), b)) \\
            &= \max \{ \majrk(c) \mid c \in \Base_{F'}(\head (\bar a)) \cup \{b \} \} \\
            &< \majrk(\alpha_1(\bar a)),
        \end{align*}
        so there exists $c \coloneqq \normal((\alpha_0 \circ \trig_A)(\head (\bar a), b)) \in \Normal$ with $\interpr c = \interpr{\alpha_1(\bar a)}$ and $\majrk(c) < \majrk(\alpha_1(\bar a))$ (Existence, \Cref{prop:existenceOfNormalTerms}). This contradicts normality of $\alpha_1(\bar a)$. Hence there does not exists $b \in A$ with $\interpr b = \interpr{(\alpha_1 \circ \tail) (\bar a)}$ and $\majrk(b) < \majrk(\alpha_1(\bar a))$. In particular, this implies $\majrk((\alpha_1 \circ \tail) (\bar a)) \geq \majrk(\alpha_1(\bar a))$ and so $\majrk((\alpha_1 \circ \tail) (\bar a)) = \majrk(\alpha_1(\bar a))$. Furthermore, $(\alpha_1 \circ \tail) (\bar a)$ satisfies \Cref{item:normalityRank} of normality, because $\rk((\alpha_1 \circ \tail) (\bar a)) = (\majrk((\alpha_1 \circ \tail) (\bar a)), 0)$. Finally, $(\alpha_1 \circ \tail) (\bar a)$ satisfies \Cref{item:normalitySubterms} of normality, because $\subterm((\alpha_1 \circ \tail) (\bar x)) \subseteq \subterm(\alpha_1(\bar x)) \subseteq \Normal$. 
    \end{proof}
\end{toappendix}

\begin{apxpropositionrep}[Uniqueness of normal representatives]
\label{prop:uniquenessOfNormalTerms}
    For all $a, b \in \Normal$, if $\interpr a = \interpr b$ then $a = b$.
\end{apxpropositionrep}
\begin{proofsketch}
    We prove by induction on $(i, j)$ that for all $a, b \in \Rk_{(i, j)} \cap \Normal$, $\interpr a = \interpr b$ implies $a = b$. Note that either $a, b \in \Branch$ or $a, b \in \Stream$. In the first case, by the induction hypothesis, normal subterms of $a$ and $b$ with the same semantics are unique, which can be shown to imply $a = b$. In the latter case, $a = \alpha_1(\bar a)$ and $b = \alpha_1(\bar b)$ for $\bar{a},\bar{b} \in GA = (F'A)^\omega$. 
    We prove $\head(\bar a) = \head(\bar b)$ and $\interpr{\tail(\bar a)} = \interpr{\tail(\bar b)}$. So by coinduction on streams, $\bar a = \bar b$.
\end{proofsketch}
\begin{proof}
    We prove by induction on $(i, j)$ that if $a, b \in \Normal$ and $a \in \Rk_{(i,j)}$, then $a = b$. Suppose the property holds for $\Rk_{\prec(i,j)}$ and let $a, b \in \Normal$ with $a \in \Rk_{(i,j)}$. By \Cref{item:normalityRank} of normality, $\rk(b) = \rk(a) = (i,j)$. We consider the following cases.
    \begin{itemize}
        \item $a = \alpha_0(\bar a)$, $b = \alpha_1(\bar b)$, for some $\bar a \in FA$, $\bar b \in GA$. We show this case is impossible. From $\rk(a) = \rk(b)$ and $b \in \Stream$, it follows that $\minrk(a) = \minrk(b) = 0$. But since $a \in \Branch$, $\minrk(a) \geq 1$, which is a contradiction.
        \item $a = \alpha_1(\bar a)$, $b = \alpha_0(\bar b)$, for some $\bar a \in GA$, $\bar b \in FA$. This case is symmetric to the previous one.
        \item $a = \alpha_0(\bar a)$, $b = \alpha_0(\bar b)$, for some $\bar a, \bar b \in FA$. We have $F\interpr-(\bar a) = (F\interpr- \circ \epsilon)(a) = (F\interpr- \circ \epsilon)(b) = F\interpr-(\bar b)$. By the induction hypothesis, $\interpr-$ is monic on $\Base_F(\bar a) \cup \Base_F(\bar b) \subseteq \Rk_{\prec(i,j)} \cap \Normal$, so $F\interpr-$ is monic on $F(\Base_F(\bar a) \cup \Base_F(\bar b)) \ni \bar a, \bar b$. Now $F\interpr-(\bar a) = F\interpr-(\bar b)$ implies $\bar a = \bar b$, thus $a = b$.

        \item $a = \alpha_1(\bar a)$, $b = \alpha_1(\bar b)$, for some $\bar a, \bar b \in GA=(F'A)^\omega$. 
        Our strategy is to show that $\interpr{\alpha_1(\bar a)} = \interpr{\alpha_1(\bar b)}$ implies $\head(\bar a) = \head(\bar b)$ and $\interpr{(\alpha_1\circ \tail) (\bar a)} = \interpr{(\alpha_1\circ \tail) (\bar b)}$. 
        In other words, we show that the kernel of $\interpr{-}\circ\alpha_1$ restricted to $GA\cap \alpha_1^{-1}(\Normal\cap \Rk_{(i,j)})$ is a stream bisimulation. Hence by coinduction, cf.~\cite[Def.2.3, Thm.2.5]{Rutten:BDE}, $\interpr{\alpha_1(\bar{a})} = \interpr{\alpha_1(\bar{b})}$ implies $\bar{a}=\bar{b}$, and hence $a = b$.
        
        We first prove that $\interpr{(\alpha_1 \circ \tail) (\bar a)} = \interpr{(\alpha_1\circ \tail) (\bar b)}$. Observe that:
        \begin{align*}
            &\mathcal P \interpr-(\Base_{F'}(\head (\bar a)) \cup \{(\alpha_1\circ \tail) (\bar a)\}) \\
            &\hspace{6px}= \mathcal P\interpr-(\Base_F(\trig_A(\head (\bar a), (\alpha_1\circ \tail) (\bar a))) \\
            &\hspace{6px}= \Base_F((F\interpr- \circ \trig_A)(\head (\bar a), (\alpha_1\circ \tail) (\bar a))) \\
            &\hspace{6px}= \Base_F((F\interpr- \circ \epsilon \circ \alpha_1)(\bar a)) \\
            &\hspace{6px}= \Base_F((\zeta \circ \interpr- \circ \alpha_1)(\bar a)) \\
            &\hspace{6px}= \Base_F((\zeta \circ \interpr- \circ \alpha_1)(\bar b)) \\
            &\hspace{6px}= \hdots \\
            &\hspace{6px}= \mathcal P \interpr-(\Base_{F'}(\head (\bar b)) \cup \{(\alpha_1\circ \tail) (\bar b)\}).
        \end{align*}
        Consequently, there exists $b_0 \in \Base_{F'}(\head(\bar b)) \cup \{(\alpha_1 \circ \tail) (\bar b)\}$ with $\interpr{b_0} = \interpr{(\alpha_1 \circ \tail) (\bar a)}$. By \Cref{lem:tailIsNormalWithSameMajRank} and \Cref{item:normalitySubterms} of normality, we know $b_0 \in \Normal$.
        Again by \Cref{lem:tailIsNormalWithSameMajRank}, $\majrk((\alpha_1 \circ \tail) (\bar a)) = \majrk(\alpha_1(\bar a)) = i$ and $(\alpha_1 \circ \tail) (\bar a) \in \Normal$. Hence, by \Cref{item:normalityRank} of normality, $\majrk((\alpha_1 \circ \tail) (\bar a)) = \majrk(b_0) = i$. Since every $b_1 \in \Base_{F'}(\head(\bar b)))$ has $\majrk(b_1) < \majrk(\alpha_1(\bar b)) = i$, this implies $b_0 = (\alpha_1 \circ \tail) (\bar b)$, i.e., $\interpr{(\alpha_1 \circ \tail) (\bar a)} = \interpr{(\alpha_1 \circ \tail) (\bar b)}$.

        Next, we prove that $\head(\bar a) = \head(\bar b)$. Observe that:
        \begin{align*}
            &\hspace{-16px}(F\interpr- \circ \trig_A)(\head(\bar a), (\alpha_1 \circ \tail) (\bar a)) \\
            &= \hdots \\
            &=  (\zeta \circ \interpr- \circ \alpha_1)(\bar a) \\
            &= (\zeta \circ \interpr- \circ \alpha_1)(\bar b) \\
            &= (F\interpr- \circ \trig_A)(\head(\bar b), (\alpha_1 \circ \tail) (\bar b)) \\
            &= \trig_Z((F'\interpr-\circ \head)(\bar b), \interpr{(\alpha_1 \circ \tail) (\bar b)}) \\
            &= \trig_Z((F'\interpr-\circ \head)(\bar b), \interpr{(\alpha_1 \circ \tail) (\bar a)}) \\
            &= (F\interpr- \circ \trig_A)(\head(\bar b), (\alpha_1 \circ \tail) (\bar a)),
        \end{align*}
        By the induction hypothesis, $\interpr-$ is monic on $\Base_{F'}(\head(\bar a)) \cup \Base_{F'}(\head(\bar b))$ $\subseteq Maj_{<i} \cap \Normal$. Furthermore, the fact that $(\alpha_1 \circ \tail) (\bar a)$ is normal of major rank $i$ ensures that $\interpr{(\alpha_1 \circ \tail) (\bar a)} \neq \interpr c$ for any $c \in \Base_{F'}(\head(\bar a)) \cup \Base_{F'}(\head(\bar b))$. Therefore $\interpr-$ is also monic on $\Base_{F'}(\head(\bar a)) \cup \Base_{F'}(\head(\bar b))) $ $\cup \{(\alpha_1 \circ \tail) (\bar a)\}$. Now:
        \begin{itemize}
            \item $F\interpr-$ is monic on $B \coloneqq F(\Base_{F'}(\head(\bar a)) \cup \Base_{F'}(\head(\bar b))) \cup \{(\alpha_1 \circ \tail) (\bar a)\})$;
            \item $\trig_A(\head(\bar a), (\alpha_1 \circ \tail) (\bar a)) \in B$;
            \item $\trig_A(\head(\bar b), (\alpha_1 \circ \tail) (\bar a)) \in B$;
            \item $(F\interpr- \circ  \trig_A)(\head(\bar a), (\alpha_1 \circ \tail) (\bar a)) = (F\interpr- \circ \trig_A)(\head(\bar b), (\alpha_1 \circ \tail) (\bar a))$.
        \end{itemize}
        Thus $\trig_A(\head(\bar a), (\alpha_1 \circ \tail) (\bar a)) = \trig_A(\head(\bar b), (\alpha_1 \circ \tail) (\bar a))$. But since $(\alpha_1 \circ \tail) (\bar a) \notin \Base_{F'}(\head(\bar a))$, \Cref{prop:elemPropertiesOfDerivative} (iii) implies that $\head(\bar a) = \head(\bar b)$. \qedhere
    \end{itemize}
\end{proof}

\begin{corollary}\label{cor:normal-surj}
For all $a \in \Normal \subseteq A$, $\normal(a)=a$. Hence $\normal \colon A \to \Normal$ is surjective.    
\end{corollary}
\begin{proof}
    Let $a \in \Normal$. By \Cref{prop:existenceOfNormalTerms}, $\interpr{\normal(a)} = \interpr{a}$, and by \Cref{prop:uniquenessOfNormalTerms}, $\normal(a)=a$. 
\end{proof}


\section{Thin and Constructible Behaviours Coincide}
\label{sec:thinCoincidesWithConstructible}

The aim of this section is to make a precise connection between thin coalgebras and constructible behaviours. Namely, we show in \Cref{thm:thinIFFConstructible} that an $F$-coalgebra element is thin if and only if its behaviour is constructible. This gives us two perspectives on the thinness property: the first is combinatorial and comes from directly interpreting the definition of infinite paths; the second is structural, it tells us how thin behaviours can be constructed from our syntax.

By \Cref{prop:morphismsPreservePathCount}, we have that an element in an arbitrary $F$-coalgebra is thin precisely when its behaviour is thin. Hence it suffices to prove that thin behaviours coincide with constructible behaviours. We begin by showing that constructible behaviours are thin.

\begin{apxlemmarep}
\label{lem:successorInContextOrMult1}
    Let $(T,\tau)$ be an $F$-coalgebra, $t\in T$ and $\tau(t) = \trig_T(\bar t', t_0)$. 
    For all $(s, k) \in \Suc(t)$: $s \in \Base_{F'}(\bar t')$ or $k = 0$.
\end{apxlemmarep}
\begin{proof}
    Recall that, by \Cref{prop:elemPropertiesOfDerivative} (i), $\Base_F(\tau(t)) = \Base_{F'}(\bar t') \cup \{ t_0 \}$. 
    Hence, for every successor $s$ of $t$, if $s \notin \Base_{F'}(\bar t')$ then $s=t_0$ and $s$ has multiplicity 1.   
\end{proof}

\begin{proposition}
\label{prop:thinImpliesCountableTraces}
   If $z \in Z$ is constructible, then $z$ is thin.
\end{proposition}
\begin{proof}
    By definition, the constructible behaviour $z$ has a representative $a\in A$. We prove by induction on the rank of $a\in A$ that $\InfPath(\interpr a)$ is countable. Suppose $a \in \Rk_{(i,j)}$ and that the property holds for all terms in $\Rk_{\prec(i,j)}$.
    \begin{itemize}
        \item Case $a = \alpha_0(\bar a)$ for some $\bar a \in FA$. By the definition of rank, we have $\bar a \in F(\Rk_{\prec(i,j)})$, so the induction hypothesis holds for all terms in $\Base_F(\bar a)$. Since $F\interpr-(\bar a) = \zeta(z)$, for every successor $z'$ of $z$, there exists $b \in \Base_F(\bar a)$ with $z' = \interpr{b}$. Hence, by the induction hypothesis, $\InfPath(z')$ is countable. Now:
        \begin{equation*}
            \InfPath(z) = \bigcup_{(y,k) \in \Suc(z)} (z k) \cdot \InfPath(y)
        \end{equation*}
        is a countable union of countable sets, because $\Suc(z)$ is countable (this follows from the fact that analytic functors are finitary). Therefore $z$ is thin.
        \item Case $a = \alpha_1(\bar a)$ for some $\bar a \in GA$. By the definition of rank, we have $\bar a \in G(\Rk_{\prec(i,j)})$, so the induction hypothesis holds for all terms in $\Base_G(\bar a)$. We have:
        \begin{equation*}
            z = \interpr a = (\interpr- \circ \alpha_1)(\bar a) = (\beta_1 \circ G\interpr-) (\bar a).
        \end{equation*}
        Let $\bar z \coloneqq G\interpr-(\bar a)$, $\bar z'_n \coloneqq (\head \circ \tail^{n})(\bar z)$ and $z_n \coloneqq (\beta_1 \circ \tail^n)(\bar z)$ for every $n \in \omega$. From the definition of $\beta$ it follows that:
        \begin{equation*}
            \zeta(z_n) = \trig_Z(\bar z'_n, z_{n+1}).
        \end{equation*}
        We write $(z_n0)_{n \in \omega}$ for the path $z_0 0 z_1 0 z_2 \ldots$.
        It follows from \Cref{lem:successorInContextOrMult1} that:
        \begin{align*}
            &\InfPath(z) = \{ (z_n0)_{n \in \omega} \} \: \cup \\
            &\hspace{10px} \bigcup_{n \in \omega} \bigcup_{\substack{(y, k) \in \Suc(z_n) \\ y \in \Base(\bar z'_n)}} (z_m 0)_{m < n} \cdot (z_nk) \cdot \InfPath(y).
        \end{align*}
        This is because every infinite path from $z$ is either equal to $(z_n0)_{n \in \omega}$, or it diverges from it after $n$ many successors by going to $y \in \Base_{F'}(\bar z'_n)$. Now, by the inductive hypothesis, all $\InfPath(y)$ in the above union are countable. Since all $\Base_{F'}(\bar z'_n)$ are countable, we obtain $\InfPath(z)$ as a countable union of countably many paths. Therefore $z$ is thin. \qedhere
    \end{itemize}
\end{proof}

Conversely, we show by contraposition that thin elements of the final $F$-coalgebra are constructible.
That is, we show that if $z \in Z$ is not constructible, then there are uncountably many infinite paths starting from $z$.
We show this by proving that
$\InfPath(z)$ contains a path structure similar to the full binary tree.
This path structure will be composed from infinitely many finite paths. We begin with a technical lemma.

\begin{apxlemmarep}
\label{lem:sequencesOfContextsInZ}
    If $(z_m)_{m\in\omega} \in Z^\omega$ and $(\bar z'_{m})_{m > 0} \in (F'Z)^\omega$ satisfy $\trig_Z(\bar z'_{m+1},z_{m+1}) = \zeta(z_m)$, then $\beta_1((\bar z'_m)_{m > 0}) = z_0$.
\end{apxlemmarep}
\begin{proofsketch}
    One can show that the relation $\Id_Z \cup \{ (z_n , \beta_1( (\bar z_m')_{m > n})) \mid n \in \omega\}$ on $Z$ is an $F$-bisimulation. The desired equality $z_0 = \beta_1((\bar z_m')_{m > 0})$ then follows from the fact that $(Z,\zeta)$ is a final $F$-coalgebra.
\end{proofsketch}
\begin{proof}
    We show that there exists an $F$-coalgebra coalgebra $(R, \rho)$ such that $R \subseteq Z \times (Z+GZ)$, $(z_0, \inj2((\bar z_m')_{m > 0})) \in R$ and the following two squares commute:
    \[\begin{tikzcd}
	Z & R & & Z \\
	FZ & FR & & FZ
	\arrow["\zeta"', from=1-1, to=2-1]
	\arrow["\prj1"', from=1-2, to=1-1]
	\arrow["{[\id,\beta_1] \circ \prj2}", from=1-2, to=1-4]
	\arrow["\rho", from=1-2, to=2-2]
	\arrow["\zeta", from=1-4, to=2-4]
	\arrow["{F\prj1}", from=2-2, to=2-1]
	\arrow["{F([\id,\beta_1] \circ \prj2)}"', from=2-2, to=2-4]
    \end{tikzcd}\]
    Then the desired equality $z_0 = \beta_1((\bar z_m')_{m > 0})$ will follow from the fact that $(Z,\zeta)$ is a final $F$-coalgebra.
    
    We define $(R, \rho)$ as follows:
    \begin{align*}
        &R \coloneqq (Z \times \inj1[Z]) \cup \{ (z_n, \inj2((\bar z_m')_{m > n})) \mid n \in \omega \}, \\
        &\rho(z, \inj1(z)) \coloneqq (F(\langle \id, \inj1 \rangle) \circ \zeta)(z), \\
        &\rho(z_n, \inj2((\bar z_m')_{m > n})) \coloneqq \trig_R(F'(\langle \id, \inj1 \rangle)(\bar z_{n+1}'),\\ &\hspace{118px}(z_{n+1}, \inj2((\bar z_m')_{m > n+1}))).
    \end{align*}
    To show that the two squares commute, firstly, let $z \in Z$. We have:
    \begin{multline*}
        (F\prj1 \circ \rho)(z, \inj1(z)) = (F\prj1 \circ F(\langle id, \inj1 \rangle) \circ \zeta)(z) = \\
        \zeta(z) = (\zeta \circ \prj1)(z, \inj1(z)),
    \end{multline*}
    \begin{multline*}
        (F([\id, \beta_1] \circ \prj2) \circ \rho)(z, \inj1(z)) = \\
        (F([\id, \beta_1]\circ \prj2) \circ F(\langle id, \inj1 \rangle) \circ \zeta)(z) \\
        = \zeta(z) = (\zeta \circ [\id, \beta_1]\circ \prj2)(z, \inj1(z)).
    \end{multline*}
    Secondly, let $n \in \omega$. We have:
    \begin{align*}
        &(F\prj1 \circ \rho)(z_n, \inj2((\bar z_m')_{m > n})) \\
        &\hspace{15px} = (F\prj1 \circ \trig_R)(F'(\langle \id, \inj1 \rangle)(\bar z_{n+1}'), \\
        & \hspace{83px} (z_{n+1}, \inj2((\bar z_m')_{m > n+1}))) \\
        &\hspace{15px} = \trig_Z((F'\prj1 \circ F'\langle \id, \inj1 \rangle)(\bar z_{n+1}'), \\
        & \hspace{48px} \prj1((z_{n+1}, \inj2((\bar z_m')_{m > n+1})))) \\
        &\hspace{15px} = \trig_Z(\bar z_{n+1}', z_{n+1}) \\
        &\hspace{15px} = \zeta(z_n) \hspace{20px} \text{(by assumption)}\\
        &\hspace{15px} = (\zeta \circ \prj1)((z_{n}, \inj2((\bar z_m')_{m > n}))),
    \end{align*}
    \begin{align*}
        &(F([\id, \beta_1] \circ \prj2) \circ \rho)(z_n, \inj2((\bar z_m')_{m > n})) \\
        & \hspace{15px} = (F[\id, \beta_1] \circ F\prj2 \circ \trig_R)(F'(\langle \id, \inj1 \rangle)(\bar z_{n+1}'), \\
        & \hspace{15px}  \hspace{115px} (z_{n+1}, \inj2((\bar z_m')_{m > n+1}))) \\
        & \hspace{15px} = \trig_Z(\bar z_{n+1}', \beta_1((\bar z_m')_{m > n+1})) \\
        & \hspace{15px} = (\zeta \circ \beta_1)((\bar z_m')_{m > n})  \\
        & \hspace{15px} = (\zeta \circ [\id, \beta_1] \circ \prj2)(z_n, \inj2((\bar z_m')_{m > n})). 
    \end{align*}
    where the third equality uses coherence of $(Z,\beta)$ and $\beta_0 = \zeta^{-1}$.
\end{proof}

The next lemma shows that from a non-constructible element, there are two distinct finite paths to non-constructible elements.

\begin{lemma}
\label{lem:twoNonThinChildren}
    Let $z \in Z$ be a non-constructible behaviour, i.e., $z \notin \thin Z$. Then there exists a finite path from $z$ to some $z_0$ and two distinct pairs $(z_1,k_1), (z_2,k_2) \in \Suc(z_0)$ with $z_1, z_2 \notin \thin Z$.
\end{lemma}
\begin{proof}
    Assume towards a contradiction that $z \notin \thin Z$ and for every finite path from $z$ to some $z_0$, there is at most one pair $(z_1, k_1) \in \Suc(z_0)$ with $z_1 \notin \thin Z$. We obtain an infinite path $(z_0k_1z_1 \dotsc) \in \InfPath(z)$, with $z_n \notin \thin Z$ for all $n \in \omega$, recursively as follows:
    \begin{itemize}
        \item $z_0 \coloneqq z$.
        \item 
        Suppose $z_n$ is defined and $z_n \notin \thin Z$. By assumption, $\Suc(z_n)$ contains at most one pair $(y, l)$ with $y \notin \thin Z$. 
        We show that there exists exactly one such pair. Indeed, if we assumed that all successors of $z_n$ are constructible, it would imply $\zeta(z_n) \in F(\thin Z)$, so by fixing some $f: \thin Z \to A$ with $\interpr- \circ f = \id$, we would get:
        \begin{multline*}
            (\interpr- \circ \alpha_0 \circ Ff \circ \zeta)(z_n) = \\
            (\beta_0 \circ F\interpr- \circ Ff \circ \zeta)(z_n) =
            (\beta_0 \circ \zeta)(z_n) = z_n,
        \end{multline*}
        where the first equality uses the fact that $\interpr-$ is an $(F+G)$-algebra morphism, and the third equality, that $\beta_0 = \zeta^{-1}$. But this would contradict $z_n \notin \thin Z = \interpr A$. Hence $\Suc(z_n)$ contains exactly one pair $(y, l)$ with $y \notin \thin Z$. So, necessarily, $l = 0$. We set $k_{n+1} \coloneqq 0$
        and $z_{n+1} \coloneqq y$.
    \end{itemize}
    Now for every $n \in \omega$, we have $z_{n+1} \in \Base_F(\zeta(z_n))$, hence (by \Cref{prop:elemPropertiesOfDerivative} (ii)) there exists $\bar z'_{n+1} \in F'Z$ with $\trig_Z(\bar z'_{n+1},z_{n+1}) = \zeta(z_n)$. 
    Notice that from the assumption, it follows that $z_{n+1} \notin \Base_{F'}(\bar z'_{n+1})$. 
    Indeed, if we assumed otherwise, this would imply that $z_{n+1}$ is a successor of $z_n$ with multiplicity at least $2$, thus $(z_{n+1}, 0), (z_{n+1}, 1) \in \Suc(z_n)$, contradicting our assumption about $z_n$. 
    By a similar argument, for all $y \in \thin Z$, we have $y \notin \Base_{F'}(\bar z'_{n+1})$. Therefore
    $\bar z'_{n+1} \in F'\thin Z$ for all $n \in \omega$, and so $(\bar z'_n)_{n\in\omega} \in G \thin Z$. Now, by \Cref{lem:sequencesOfContextsInZ}, $\beta_1((\bar z'_{n})_{n\in\omega}) = z_0$. Hence:
    \begin{multline*}
        (\interpr- \circ \alpha_1 \circ Gf)((\bar z'_n)_{n\in\omega}) = \\
        (\beta_1 \circ G\interpr- \circ  Gf)((\bar z'_n)_{n\in\omega}) =
        \beta_1((\bar z'_n)_{n\in\omega}) = z,
    \end{multline*}
    which contradicts $z_0 = z \notin \thin Z$.
\end{proof}

We can now show that the infinite paths of a non-constructible element essentially contain the full binary tree.

\begin{proposition}
\label{prop:countableTracesImpliesThin}
    If $z \in Z$ is thin, then $z$ is constructible.
\end{proposition}
\begin{proof}
    We reason by contraposition. Suppose $z \notin \thin Z$. By \Cref{lem:twoNonThinChildren}, there exists $z_0, z_1, z_2 \in Z$ such that there is a path from $z$ to $z_0$, and $(z_1,k_1),(z_2,k_2) \in \Suc(z_0)$ are distinct pairs with $z_1, z_2 \notin \thin Z$. Hence, there exist two finite paths, from $z$ to $z_1$ and from $z$ to $z_2$, respectively, such that neither path is a prefix of the other. By applying the same lemma again at $z_1$ and at $z_2$, we get four finite paths: 1) $z \dotsc z_1 \dotsc z_{11}$, 2) $z \dotsc z_1, \dotsc z_{12}$, 3) $z \dotsc z_2 \dotsc z_{21}$ and 4) $z \dotsc z_2 \dotsc z_{22}$, for some $z_{11}, z_{12}, z_{21}, z_{22} \notin \thin Z$. Again, none of these paths is a prefix of any other path. After $\omega$ steps, we obtain uncountably many distinct infinite paths from $z$.
\end{proof}

By combining the two propositions above, we arrive at the correspondence between thinness and constructibility.

\begin{theorem}
\label{thm:thinIFFConstructible}
    Let $(T,\tau)$ be an $F$-coalgebra and $t \in T$. Then $t$ is thin in $(T,\tau)$ if and only if its behaviour is constructible.
\end{theorem}
\begin{proof}
    By \Cref{prop:morphismsPreservePathCount}, $t$ is thin in $(T,\tau)$ if and only if $\beh_{(T,\tau)}(t)$ is thin in $(Z,\zeta)$. By \Cref{prop:thinImpliesCountableTraces,prop:countableTracesImpliesThin}, the latter is equivalent to $\beh_{(T,\tau)}(t)$ being constructible.
\end{proof}
\Cref{thm:thinIFFConstructible} together with \Cref{def:thinCoalgebra} and \Cref{prop:morphismsPreservePathCount} justify the following definition.
\begin{definition}
\label{def:thinRankCoalg}
Let $(T,\tau)$ be an $F$-coalgebra and $t \in T$ be thin.
The \emph{rank} of $t$, denoted $\rk(t)$, is the rank of the normal representative of its behaviour $\beh_{(T,\tau)}(t)$.
\end{definition}

The above definition provides a measure of thinness of a thin element in an $F$-coalgebra. This opens the possibility of proving properties of thin elements by induction on their rank.


\section{Axiomatisation of Constructible Behaviours}
\label{sec:main-results}
 
In this section, we show that the equation \eqref{eq:plugging} from \Cref{def:coherence} axiomatises constructible behaviours.
In particular, we show that $(\thin Z, \thin \beta)$ is isomorphic to the quotient of $(A,\alpha)$ by the least congruence $\quotrel$ containing \eqref{eq:plugging}.
We follow the following strategy:
\begin{itemize}
    \item We prove that the congruence $\quotrel$
    is \emph{sound} for the semantics $\interpr-$, i.e., $a \quotrel b$ implies $\interpr a = \interpr b$ (\Cref{thm:soundnessOfQuot}). The proof relies on two facts: 1) the quotient of $(A,\alpha)$ by \Cref{eq:plugging} is an initial coherent algebra (\Cref{prop:quotient}), and 2) $(Z,\beta)$ is a coherent algebra (\Cref{lem:Z-beta-coh}).

    \item We show that the congruence $\quotrel$ 
    is \emph{complete} for the semantics $\interpr-$, i.e., $\interpr a = \interpr b$ implies $a \quotrel b$. The proof relies on the property that every term $a \in A$ can be \emph{normalised}, i.e., $a \quotrel \normal(a)$ (\Cref{prop:normalisation}). This implies that if $\interpr a = \interpr b$, then $a$ and $b$ normalise to the same normal term.
\end{itemize}

By using soundness and completeness of the quotient together with the equivalence between constructible and thin behaviours (\Cref{thm:thinIFFConstructible}), we arrive at our main result: thin behaviours admit an inductive structure (\Cref{thm:thinImageAndQuotientIsomorphic}).

The quotient of $(A,\alpha)$ by \Cref{eq:plugging} is formally defined as a coequalizer in $\Alg(F+G)$, in the following way.
For a set $X$, we denote by $\Free(X)$ the free $(F+G)$-algebra over $X$, and
for a function $f\colon X \to C$, where $C$ is the carrier of an $(F+G)$-algebra $(C,\gamma)$, 
we denote the free extension of $f$ by $\frext{f} \colon \Free(X) \to (C,\gamma)$.
We note that free $(F+G)$-algebras exist due to \cite[Theorem~6.10]{AdamekMiliusMoss2018FixedPointsOfFunctors}.

\def\evC{\ev_C}

\begin{definition}
\label{def:equationQuotient}
Denote by $\quotmap: (A, \alpha) \to (\Aquot, \alphaquot)$ the coequaliser of $\frext{\alpha_1}$ and $\frext{(\alpha_0 \circ \branch_\alpha)}$ in $\Alg(F+G)$, as shown below. For $a, b \in A$, we write $a \quotrel b$ if $\quotmap(a) = \quotmap(b)$.
\begin{center}
\begin{tikzcd}[column sep=4em]
	{\Free(GA)} & (A,\alpha) & {(\Aquot, \alphaquot)}
	\arrow["\frext{\alpha_1}", shift left, from=1-1, to=1-2]
	\arrow["\frext{(\alpha_0 \circ \branch_\alpha)}"', shift right, from=1-1, to=1-2]
	\arrow["\quotmap", from=1-2, to=1-3]
\end{tikzcd}
\end{center}
\end{definition}

\begin{proposition}\label{prop:quotient}   
The quotient $(\Aquot, \alphaquot)$ is an initial coherent $(F+G)$-algebra.
\end{proposition}

\begin{proof}
We first show that the quotient is coherent. To see this we calculate using \Cref{lem:branchFG}:
\begin{align*}
    \alphaquot_0 \circ \branch_{\alphaquot} \circ G \quotmap  = \alphaquot_0 \circ F \quotmap \circ \branch_{\alpha} = \quotmap \circ \alpha_0 \circ \branch_\alpha  = \\
    \quotmap \circ \alpha_1 = \alphaquot_1 \circ G \quotmap. 
\end{align*}
Because $\quotmap$ is surjective, we have $G\quotmap$ is surjective (all $\Set$-functors preserve surjective maps), and thus we can conclude that $\alphaquot_0 \circ \branch_{\alphaquot} = \alphaquot_1$ as required. 

Now let $(C,\gamma)$ be a coherent $(F+G)$-algebra, i.e., $\gamma_1 = \gamma_0 \circ \branch_\gamma$.
We will show that the initial $(F+G)$-algebra morphism
$\evC \colon (A,\alpha) \to (C,\gamma)$ coequalizes $\frext{\alpha_1}$ and $\frext{(\alpha_0 \circ \branch_\alpha)}$, i.e.,
$\evC \circ \frext{\alpha_1} = \evC \circ \frext{(\alpha_0 \circ \branch_\alpha)}$.
Recall that $\branch_\alpha = \trig_A \circ \tup{\head,\alpha_1 \circ \tail}$.
Consider the following diagram:\\
\centerline{$\xymatrix@C=5.5em{
A \ar@/^1.5pc/^{\evC}[rrr]
& GA \ar_-{\alpha_1}[l] \ar^-{G\evC}[r]  \ar_-{\branch_\alpha}[d]
& GC  \ar^-{\gamma_1}[r] \ar^-{\branch_\gamma}[d]
& C 
\\
& FA  \ar^-{\alpha_0}[ul] \ar[r]^-{F\evC}
& FC \ar_-{\gamma_0}[ur] 
&
}$}\\[.3em]
The middle square commutes by \Cref{lem:branchFG}, the triangle on the right commutes since $(C,\gamma)$
is coherent, and the outer part commutes since $\evC$ is an $(F+G)$-algebra morphism.
(The left triangle need not commute.)
It follows that $\evC \circ {\alpha_1} = \evC \circ {(\alpha_0 \circ \branch_\alpha)}$, that is, all
generator pairs of the congruence are identified by $\evC$.
By the uniqueness of free extensions,
it follows that
$\evC \circ \frext{\alpha_1} = \evC \circ \frext{(\alpha_0 \circ \branch_\alpha)}$.
Hence $\evC\colon (A,\alpha) \to (C,\gamma)$ is a competitor to the coequalizer $\quotmap\colon (A,\alpha) \to (\Aquot,\alphaquot)$, so by the universal property, we obtain a unique $(F+G)$-algebra morphism
$\ol{\evC} \colon (\Aquot,\alphaquot) \to (C,\gamma)$ such that $\ol{\evC} \circ \quotmap = \evC$. 
Since $\evC$ is the unique initial morphism and $\quotmap$ is epic, $\ol{\evC}$ is the unique $(F+G)$-algebra morphism from 
$(\Aquot,\alphaquot)$ to $(C,\gamma)$.
\begin{center}
    \begin{tikzcd}
        {Free(GA)} & {(A,\alpha)} & {(\Aquot,\alphaquot)} \\
        && (C,\gamma)
        \arrow["{\frext{(\alpha_0 \circ\branch_\alpha)}}"', shift right, from=1-1, to=1-2]
        \arrow["{\frext{\alpha_1}}", shift left, from=1-1, to=1-2]
        \arrow["\quotmap", from=1-2, to=1-3]
        \arrow["\evC"', from=1-2, to=2-3]
        \arrow["{\overline{\evC}}", from=1-3, to=2-3]
    \end{tikzcd}
    \qedhere
\end{center}
\end{proof}

We proceed with soundness of the quotient, which means that $\quotrel$ only relates terms with the same semantics.

\begin{theorem}[Soundness of the quotient]
\label{thm:soundnessOfQuot}
    The semantics map factors uniquely through the quotient as shown below.
    Hence, for all $a,b \in A$,  $a \quotrel b$ implies $\interpr a = \interpr b$.
    \begin{center}
    \begin{tikzcd}
        {(A,\alpha)} & {(\Aquot,\alphaquot)} & {( Z, \beta)}
        \arrow["\quotmap", two heads, from=1-1, to=1-2]
        \arrow["{\interpr-}"', curve={height=12pt}, from=1-1, to=1-3]
        \arrow["{\overline{\interpr-}}", from=1-2, to=1-3]
    \end{tikzcd}
    \end{center}
\end{theorem}  
\begin{proof}
  Immediate from the initiality of $(\Aquot,\alphaquot)$ and \Cref{lem:Z-beta-coh}. 
\end{proof}

We now turn to proving completeness. Here it is useful to think of \Cref{eq:plugging} as a conversion rule between terms. In this context, the statement $a \quotrel b$ conveys that $a$ can be converted to $b$ by multiple (possible infinitely many) applications of \Cref{eq:plugging}. The following key property states that every term can be \emph{normalised}, i.e., converted to its corresponding normal term.

\begin{toappendix}
    \paragraph*{Normalisation}
    We begin with one case where normalisation is simpler. This is when the term is an $F$-term and its subterms are normal.

    \begin{lemma}
    \label{lem:branchingWithNormalSubtermsIsNormalisable}
        If $\bar a \in F(\Normal)$, then $\alpha_0(\bar a) \quotrel (\normal \circ \alpha_0)(\bar a)$.
    \end{lemma}
    \begin{proof}
        Firstly, we claim that $(\normal\circ \alpha_0)(\bar a) \quotrel \alpha_0(\bar b)$ for some $\bar b \in F(\Normal)$. If $(\normal \circ \alpha_0)(\bar a) \in \Branch$, then $(\normal\circ \alpha_0)(\bar a) = \alpha_0(\bar b)$ for some $\bar b \in F(\Normal)$, so we are done. Otherwise, $(\normal\circ \alpha_0)(\bar a) = \alpha_1(\bar c)$ for some $\bar c \in G(\Normal)$. We have:
        \begin{multline*}
            (\normal \circ \alpha_0)(\bar a) = \alpha_1(\bar c) \quotrel (\alpha_0 \circ \branch_\alpha)(\bar c) = \\
            (\alpha_0 \circ \trig_A)(\head(\bar c), (\alpha_1 \circ \tail) (\bar c)),
        \end{multline*}
        with $\Base_F(\trig_A(\head(\bar c), (\alpha_1 \circ \tail) (\bar c))) = \Base_{F'}(\head (\bar c)) \cup \{ (\alpha_1 \circ \tail) (\bar c) \}) \subseteq \Normal$, where the last inclusion follows from \Cref{lem:tailIsNormalWithSameMajRank}. Therefore we can take $\bar b \coloneqq \trig_A(\head(\bar c), (\alpha_1 \circ \tail) (\bar c))$. This proves the claim.
    
        By applying soundness of $\quotrel$ (\Cref{thm:soundnessOfQuot}), we get:
        \begin{equation*}
            \interpr{\alpha_0(\bar b)} = \interpr{(\normal \circ \alpha_0)(\bar a)} = \interpr{\alpha_0(\bar a)}
        \end{equation*}
        Therefore $F\interpr-(\bar a) = F\interpr-(\bar b)$. We know by Uniqueness of normal representatives (\Cref{prop:uniquenessOfNormalTerms}) that $\interpr-$ is monic on $\Normal$. Now from the fact that $F$ preserves monos, we deduce $\bar a = \bar b$. Hence:
        \begin{equation*}
            \alpha_0(\bar a) = \alpha_0(\bar b) \quotrel (\normal \circ \alpha_0)(\bar a).
        \end{equation*}
        The statement now follows from reflexivity and transitivity of $\quotrel$. 
    \end{proof}

    Next, we turn to the $G$-term case. Given a G-term $a$, we first develop useful lower bounds for $\rk(\normal(a))$ and $\majrk(\normal(a))$.

    \begin{lemma}
    \label{lem:normalRanksRelations}
        Let $\bar a \in G(\Normal)$.
        \begin{enumerate}[(a)]
            \item\label{item:normalTailInequality} $\rk((\normal \circ \alpha_1 \circ \tail)(\bar a)) \preceq \rk((\normal \circ \alpha_1)(\bar a))$. The inequality is strict if $n \circ \alpha_1(\bar a) \in \Branch$.

            \item\label{item:normalHeadInequality} $\sup \{ \rk(b) \mid b \in \Base_{F'}(\head(\bar a)) \} \preceq \rk((\normal \circ \alpha_1)(\bar a))$.
            \item\label{item:normalTailMajorRankInequality} $\majrk((\normal \circ \alpha_1 \circ \tail)(\bar a)) < \majrk((\normal \circ \alpha_1)(\bar a))$, if $(\normal \circ \alpha_1)(\bar a) = \alpha_1(\bar b)$, for some $\bar b \in GA$ with $\head(\bar a) \neq \head(\bar b)$.
        \end{enumerate}
    \end{lemma}
    \begin{proof}
        (\Cref{item:normalTailInequality}) Suppose $(\normal \circ \alpha_1)(\bar a) = \alpha_0(\bar b)$ for some $\bar b \in FA$. Since $\interpr{\alpha_0(\bar b)} = \interpr{\alpha_1(\bar a)}$:
        \begin{multline*}
            F\interpr-(\bar b) = (F\interpr- \circ \epsilon \circ \alpha_0)(\bar b) = (F\interpr- \circ \epsilon \circ \alpha_1)(\bar a) \\
            = (F\interpr- \circ \trig_A)(\head(\bar a), (\alpha_1 \circ \tail)(\bar a)).
        \end{multline*}
        Therefore:
        \begin{align*}
            &\interpr{(\normal \circ \alpha_1 \circ \tail) (\bar a)} \\
            &\hspace{6px}= \interpr{(\alpha_1\circ \tail) (\bar a)} \\
            &\hspace{6px}\in \mathcal P\interpr-(\Base_F(\trig_A(\head(\bar a), (\alpha_1\circ \tail) (\bar a)))) \\
            &\hspace{6px}= \Base_F((F\interpr- \circ \trig_A)(\head(\bar a), (\alpha_1\circ \tail) (\bar a))) \\
            &\hspace{6px}= \Base_F(F\interpr-(\bar b)) \\
            &\hspace{6px}= \mathcal P\interpr-(\Base_F(\bar b)).
        \end{align*}
        Since $\{(\normal \circ \alpha_1 \circ \tail)(\bar a) \} \cup \Base_F(\bar b) \subseteq \Normal$, Uniqueness (\Cref{prop:uniquenessOfNormalTerms}) implies that $(\normal \circ \alpha_1 \circ \tail)(\bar a) \in \Base_F(\bar b) = \subterm((\normal \circ \alpha_1)(\bar a))$. Hence $\rk((\normal \circ \alpha_1 \circ \tail)(\bar a)) < \rk((\normal \circ \alpha_1)(\bar a))$.

        Now suppose $(\normal \circ \alpha_1)(\bar a) = \alpha_1(\bar b)$ for some $\bar b \in GA$. We have $\interpr{\trig_A(\head(\bar b), (\alpha_1 \circ \tail)(\bar b))} = \interpr{\alpha_1(\bar b)} = \interpr{\alpha_1(\bar a)}$. Analogously to the previous case, we obtain:
        \begin{multline*}
            (\normal \circ \alpha_1 \circ \tail)(\bar a) \in \Base_F(\trig_A(\head(\bar b), (\alpha_1 \circ \tail)(\bar b))) = \\ \Base_{F'}(\head(\bar b)) \cup \{ (\alpha_1 \circ \tail)(\bar b) \}.
        \end{multline*}
        If $(\normal \circ \alpha_1 \circ \tail)(\bar a) \in \Base_{F'}(\head(\bar b))$, the inequality follows immediately. If $(\normal \circ \alpha_1 \circ \tail)(\bar a) = (\alpha_1 \circ \tail)(\bar b)$, the inequality follows from \Cref{lem:tailIsNormalWithSameMajRank}.

        (\Cref{item:normalHeadInequality}) Suppose $(\normal \circ \alpha_1)(\bar a) = \alpha_0(\bar b)$ for some $\bar b \in FA$. Analogously to above:
        \begin{multline*}
            \mathcal P\interpr-(\Base_F(\bar b)) =\\
            \mathcal P \interpr-(\Base_F(\trig_A(\head(\bar a), (\alpha_1 \circ \tail)(\bar a)))) =\\
            \mathcal P \interpr- (\Base_{F'}(\head(\bar a)) \cup \{ (\alpha_1 \circ \tail)(\bar a)\}).
        \end{multline*}
        From $\Base_{F'}(\head(\bar a)) \subseteq \Normal$ and Uniqueness (\Cref{prop:uniquenessOfNormalTerms}), it follows that $\Base_{F'}(\head(\bar a))$ $\subseteq \Base_F(\bar b) = \subterm((\normal \circ \alpha_1)(\bar a))$, which implies the desired inequality.

        Now suppose $(\normal \circ \alpha_1)(\bar a) = \alpha_1(\bar b)$ for some $\bar b \in GA$. Analogously to above, we get:
        \begin{multline*}
            \mathcal P\interpr-(\Base_{F'}(\head(\bar b)) \cup \{ (\alpha_1 \circ \tail)(\bar b)\}) =\\ \mathcal P\interpr-(\Base_{F'}(\head(\bar a)) \cup \{ (\alpha_1 \circ \tail)(\bar a)\}).
        \end{multline*}
        Since $\Base_{F'}(\head(\bar a)) \subseteq \Normal$, we get $\Base_{F'}(\head(\bar a)) \subseteq \Base_{F'}(\head(\bar b)) \cup \{ (\alpha_1 \circ \tail)(\bar b)\}$, from which the inequality follows.

        (\Cref{item:normalTailMajorRankInequality}) By assumption, $(\normal \circ \alpha_1)(\bar a) = \alpha_1(\bar b)$ for some $\bar b \in GA$, and $\bar a \neq \bar b$. Since $\interpr{(\alpha_0 \circ \trig_A)(\head(\bar a), (\normal \circ \alpha_1 \circ \tail)(\bar a))} = \interpr{(\alpha_0 \circ \trig_A)(\head(\bar b), (\alpha_1 \circ \tail)(\bar b))}$, we analogously get:
        \begin{align*}
            (F\interpr- \circ \trig_A)(\head(\bar a), (\normal \circ \alpha_1 \circ \tail)(\bar a)) &= \\
            & \hspace{-120px} (F\interpr- \circ \trig_A)(\head(\bar b), (\alpha_1 \circ \tail)(\bar b)), \\
            \mathcal P\interpr-(\Base_{F'}(\head(\bar a)) \cup \{(\normal \circ \alpha_1 \circ \tail)(\bar a)\}) &= \\
            & \hspace{-120px} \mathcal P\interpr-(\Base_{F'}(\head(\bar b)) \cup \{(\alpha_1 \circ \tail)(\bar b)\}).
        \end{align*}
        By Uniqueness (\Cref{prop:uniquenessOfNormalTerms}), these equations turn into:
        \begin{align*}
            &\trig_A(\head(\bar a), (\normal \circ \alpha_1 \circ \tail)(\bar a)) = \trig_A(\head(\bar b), (\alpha_1 \circ \tail)(\bar b)), \\
            &\Base_{F'}(\head(\bar a)) \cup \{(\normal \circ \alpha_1 \circ \tail)(\bar a)\} = \Base_{F'}(\head(\bar b)) \: \cup \\
            & \hspace{160px} \{(\alpha_1 \circ \tail)(\bar b)\}.
        \end{align*}
        By the contrapositive of \Cref{prop:elemPropertiesOfDerivative} (iii), we get $(\normal \circ \alpha_1 \circ \tail)(\bar a) \neq (\alpha_1 \circ \tail)(\bar b)$, hence $(\normal \circ \alpha_1 \circ \tail)(\bar a) \in \Base_{F'}(\head(\bar b))$. This implies the desired inequality. 
    \end{proof}

    Given a $G$-term $a$, suppose we can inductively normalise terms of lower rank. We know that the subterms of $a$ have lower rank, so they can be normalised. The challenging part is how to normalise the term corresponding to the tail of the stream, because that term might be of the same rank. The following crucial lemma takes care of this by finding a stream suffix of sufficiently low major rank.

    \begin{lemma}
    \label{lem:streamTailsReachesLowMajorRank}
        If $\bar a \in G(\Normal)$ and $\majrk((\normal \circ \alpha_1)(\bar a)) = i$, then there exists $k\in \omega$ with $\majrk((\alpha_1 \circ \tail^k) (\bar a)) \leq i$.
    \end{lemma}
    \begin{proof}
        Firstly, we claim that there exist some $k_0 \in \omega$ such that $(\normal \circ \alpha_1 \circ \tail^{k_0}) (\bar a) \in \Stream$. For if $(\normal \circ \alpha_1 \circ \tail^k) (\bar a) \in \Branch$ for all $k\in\omega$, by \Cref{lem:normalRanksRelations} \Cref{item:normalTailInequality}, we would get $\rk((\normal \circ \alpha_1 \circ \tail^{k+1}) (\bar a)) \prec \rk((\normal \circ \alpha_1 \circ \tail^k) (\bar a))$ for all $k \in \omega$, so:
        \begin{multline*}
            \rk((\normal \circ \alpha_1 \circ \tail^0) (\bar a)) \succ \rk((\normal \circ \alpha_1 \circ \tail^1) (\bar a)) \succ \dotsb \succ \\
            \rk((\normal \circ \alpha_1 \circ \tail^{k}) (\bar a)) \succ \dotsb,
        \end{multline*}
        which is a contradiction with the well-foundedness of $\prec$. Thus let $(\normal \circ \alpha_1 \circ \tail^{k_0})(\bar a) = \alpha_1(\bar b)$ for some $\bar b \in GA$.
        
        Assume towards a contradiction that $\majrk((\alpha_1 \circ \tail^{k})(\bar a)) > i$ for all $k \in \omega$. By \Cref{lem:normalRanksRelations} \Cref{item:normalTailInequality}:
        \begin{multline*}
            \majrk((\normal \circ \alpha_1 \circ \tail^{k_0}) (\bar a)) \leq \majrk((\normal \circ \alpha_1 \circ \tail^{k_0-1}) (\bar a))  \\
            \leq \dotsb \leq \majrk((\normal \circ \alpha_1 \circ \tail^0) (\bar a)) = i.
        \end{multline*}
        Therefore $\bar b \neq \tail^{k_0}(\bar a)$. Without loss of generality, assume $\head(\bar b) \neq \head(\tail^{k_0}(\bar a))$. Since $\majrk((\alpha_1 \circ \tail^{k_0+1})(\bar a)) > i$, there exists $k_1 > k_0$ with:
        \begin{equation*}
            \sup\{\majrk(c) \mid c \in \Base_{F'}((\head \circ \tail^{k_1})(\bar a)\} \geq i.
        \end{equation*}
        By \Cref{lem:normalRanksRelations} \Cref{item:normalHeadInequality}, $\majrk((\normal \circ \alpha_1 \circ \tail^{k_1})(\bar a)) \geq i$. By \Cref{lem:normalRanksRelations} \Cref{item:normalTailInequality}:
        \begin{multline*}
            i \leq \majrk((\normal \circ \alpha_1 \circ \tail^{k_1}) (\bar a)) \leq \majrk((\normal \circ \alpha_1 \circ \tail^{k_1-1}) (\bar a))  \\
            \leq \dotsb \leq \majrk((\normal \circ \alpha_1 \circ \tail^{k_0+1}) (\bar a)).
        \end{multline*}
        Finally, by \Cref{lem:normalRanksRelations} \Cref{item:normalTailMajorRankInequality}, $i \leq \majrk((\normal \circ \alpha_1 \circ \tail^{k_0+1})(\bar a)) < \majrk((\normal \circ \alpha_1 \circ \tail^{k_0})(\bar a)) \leq i$, which is a contradiction. 
    \end{proof}
\end{toappendix}

\begin{apxpropositionrep}[Normalisation]
\label{prop:normalisation}
    For all $a \in A$, $a \quotrel \normal(a)$.
\end{apxpropositionrep}
\begin{proofsketch}
    We prove by induction on $(i, j)$ that for all $a \in \Rk_{(i,j)}$, we have $a \quotrel \normal(a)$. By the induction hypothesis, without loss of generality, $\subterm(a) \subseteq \Normal$, because subterms can be inductively normalised. If $a \in \Branch$, then $a \quotrel \normal(a)$ follows from properties of $\quotrel$. Otherwise, $a \in \Stream$, i.e., $a = \alpha_1(\bar a)$ for some $\bar a$. If $a$ is not already normal, it can be shown that $(\alpha_1 \circ \tail^k)(\bar a)$ has lower rank than $a$ for some $k \in \omega$, so, by the induction hypothesis, $(\alpha_1 \circ \tail^k)(\bar a) \quotrel \normal((\alpha_1 \circ \tail^k)(\bar a))$. Then, by induction on $l = k, k-1, \dotsc, 0$, it can be shown that $(\alpha_1 \circ \tail^l)(\bar a) \quotrel \normal((\alpha_1 \circ \tail^l)(\bar a))$. For $l = 0$, this means $a \quotrel \normal(a)$.
\end{proofsketch}
\begin{proof}
    We proceed by induction on $\rk(a)$. Suppose that $a \in \Rk_{(i,j)}$ and that for any $b \in \Rk_{\prec(i,j)}$ we have $b \quotrel \normal(b)$.
    \begin{itemize}
        \item Let $a = \alpha_0(\bar a)$ for some $\bar a \in FA$. By the induction hypothesis, $(\quotmap \circ \normal)(b) = \quotmap(b)$ for all $b \in \Base_F(\bar a) \subseteq \Rk_{\prec(i,j)}$. Hence:
        \begin{align*}
            \alpha_0(\bar a) &\quotrel (\alpha_0 \circ F\normal)(\bar a) \\
            &\quotrel (\normal \circ \alpha_0 \circ F\normal)(\bar a) \\
            &= (\normal \circ \alpha_0)(\bar a)
        \end{align*}
        where the first line uses \Cref{lem:subtermPreservationImpliesTermPreservation} applied with $f \coloneqq \normal$, $B \coloneqq \Rk_{\prec(i,j)}$, $g \coloneqq \quotmap$, the second line uses \Cref{lem:branchingWithNormalSubtermsIsNormalisable}.
        For the third line we observe that $\interpr{(\alpha_0 \circ F\normal)(\bar a)} = \interpr{\alpha_0(\bar a)}$ from \Cref{lem:subtermPreservationImpliesTermPreservation} with $f\coloneqq \normal$, $g \coloneqq \interpr-$, and then we apply Uniqueness of normal representatives (\Cref{prop:uniquenessOfNormalTerms}).
        
        \item Let $a = \alpha_1(\bar a)$ for $\bar a \in GA$. By the induction hypothesis, $(\quotmap \circ \normal)(b) = \quotmap(b)$ for $b \in \Base_G(\bar a) \subseteq \Majrk_{<i}$. Let $\bar b \coloneqq G\normal(\bar a)$. By applying \Cref{lem:subtermPreservationImpliesTermPreservation} with $f \coloneqq \normal$, $g \coloneq \interpr-$, we get $\interpr{\alpha_1(\bar a)} = \interpr{\alpha_1(\bar b)}$. By applying \Cref{lem:subtermPreservationImpliesTermPreservation} with $f \coloneqq \normal$, $g \coloneqq \quotmap$, we get $\alpha_1(\bar a) \quotrel \alpha_1(\bar b)$. If $\alpha_1(\bar b) \in \Normal$, we are done. Otherwise, by the definition of normality (\Cref{def:normalTerm}), $\rk((\normal \circ \alpha_1)(\bar b)) < \rk(\alpha_1(\bar b))$. And since $\minrk(\alpha_1(\bar b)) = 0$, this implies $\majrk((\normal \circ \alpha_1)(\bar b)) < \majrk(\alpha_1(\bar b)) \leq i$. By \Cref{lem:streamTailsReachesLowMajorRank}, there exists $k \in \omega$ such that $\majrk((\alpha_1 \circ \tail^k) (\bar b)) \leq \majrk((\normal \circ \alpha_1)(\bar b)) < i$.

        We prove that $(\alpha_1 \circ \tail^l) (\bar b) \quotrel (\normal \circ \alpha_1 \circ \tail^l) (\bar b)$ by induction on $l \leq k$ in descending order, i.e., for $l = k, k-1, \dotsc, 0$.
        \begin{itemize}
            \item Base case $l = k$. Since $\majrk((\alpha_1 \circ \tail^k) (\bar b)) < i$, we have $(\alpha_1 \circ \tail^k) (\bar b) \quotrel (\normal \circ \alpha_1 \circ \tail^k) (\bar b)$ by the outer induction hypothesis.
            \item Suppose $(\alpha_1 \circ \tail^{l+1}) (\bar b) \quotrel (\normal \circ \alpha_1 \circ \tail^{l+1})(\bar b)$, we show $(\alpha_1 \circ \tail^{l})(\bar b) \quotrel (\normal \circ \alpha_1 \circ \tail^{l})(\bar b)$:
            \begin{align*}
                &(\alpha_1 \circ \tail^l) (\bar b) \\
                &\hspace{4px} \quotrel (\alpha_0 \circ \trig_A)((\head \circ \tail^l) (\bar b), (\alpha_1 \circ \tail \circ \tail^l) (\bar b)) \\
                &\hspace{4px} \quotrel (\alpha_0 \circ F\normal \circ \trig_A)((\head \circ \tail^l) (\bar b), (\alpha_1 \circ \tail \circ \tail^l) (\bar b)) \\
                &\hspace{4px} \quotrel (\normal \circ \alpha_0 \circ F\normal \circ \trig_A)((\head \circ \tail^l) (\bar b),\\
                &\hspace{103px} (\alpha_1 \circ \tail \circ \tail^l) (\bar b)) \\
                &\hspace{4px} = (\normal \circ \alpha_1 \circ \tail^l) (\bar b),
            \end{align*}
            where the first line follows by Soundness of the quotient (\Cref{thm:soundnessOfQuot}), the second line -- by the inner induction hypothesis and \Cref{lem:subtermPreservationImpliesTermPreservation} with $f \coloneqq \normal$, $g \coloneqq \quotmap$, the third line -- by \Cref{lem:branchingWithNormalSubtermsIsNormalisable}, and the fourth line -- by Uniqueness of normal representatives (\Cref{prop:uniquenessOfNormalTerms}) because the last two terms have the same semantics.
        \end{itemize}
        This completes the inner induction. Consequently, for $l = 0$ we have $(\normal \circ \alpha_1)(\bar b) \quotrel \alpha_1(\bar b)$, so:
        \begin{equation*}
            (\normal \circ \alpha_1)(\bar a) = (\normal \circ \alpha_1)(\bar b) \quotrel \alpha_1(\bar b) \quotrel \alpha_1(\bar a),
        \end{equation*}
        where the (first) equality follows from $\interpr{\alpha_1(\bar a)} = \interpr{\alpha_1(\bar b)}$ and Uniqueness of normal representatives (\Cref{prop:uniquenessOfNormalTerms}). \qedhere
    \end{itemize}
\end{proof}

\begin{theorem}[Completeness of the quotient]
    The factorisation map $\overline{\interpr-}$ in the following diagram is injective. Hence, for all $a, b \in A$, $\interpr a = \interpr b$ implies $a \quotrel b$.
    \begin{center}
    \begin{tikzcd}
        {(A,\alpha)} & {(\Aquot,\alphaquot)} & {( Z, \beta)}
        \arrow["\quotmap", two heads, from=1-1, to=1-2]
        \arrow["{\interpr-}"', curve={height=12pt}, from=1-1, to=1-3]
        \arrow["{\overline{\interpr-}}", tail, from=1-2, to=1-3]
    \end{tikzcd}
    \end{center}
\label{thm:completeness}
\end{theorem}
\begin{proof}
    We first prove that $(\thin Z, \thin \beta) \cong (\Normal, \alphanormal)$ for a suitable $(F+G)$-algebraic structure $\alphanormal$. Let $\iota: \Normal \to A$ be the inclusion map. Our strategy is to show that
    \begin{tikzcd}
        {(A,\alpha)} & {(\Normal, \alphanormal)} & {(Z,\beta)}
        \arrow["\normal", two heads, from=1-1, to=1-2]
        \arrow["{\interpr- \circ \iota}", tail, from=1-2, to=1-3]
    \end{tikzcd}
    is an epi-mono factorisation of $\interpr-: (A, \alpha) \to (Z,\beta)$ in $\Alg(F+G)$, and use the fact that all factorisations are isomorphic. Since epi-mono factorisations in $\Set$ lift to $\Alg(F+G)$ (see, e.g., \cite[Lemma~3.5]{wissmann:CALCO2021:MinimalityNotionsViaFactorizationSystems}), it suffices to show that \begin{tikzcd}
	A & \Normal & Z
	\arrow["\normal", two heads, from=1-1, to=1-2]
        \arrow["{\interpr- \circ \iota}", tail, from=1-2, to=1-3]
    \end{tikzcd}
    is an epi-mono factorisation of $\interpr-$ in $\Set$.
    We have $\interpr- = \interpr- \circ \iota\circ \normal$
    from the fact that for all $a \in A$, $\interpr{\normal(a)} = \interpr{a}$ (\Cref{prop:existenceOfNormalTerms}). Moreover, $\normal$ is surjective by \Cref{cor:normal-surj}, and
    $\interpr- \circ \iota$ is injective by uniqueness of normal representatives (\Cref{prop:uniquenessOfNormalTerms}).

    Now by Soundness of the quotient (\Cref{thm:soundnessOfQuot}), we have that $\normal: (A, \alpha) \to (\Normal, \alphanormal)$ factors through $(\Aquot, \alphaquot)$ as shown below:
    \begin{center}
    \begin{tikzcd}
        {(A,\alpha)} & {(\Aquot,\alphaquot)} & {( \Normal, \alphanormal)}
        \arrow["\quotmap", two heads, from=1-1, to=1-2]
        \arrow["{\normal}"', curve={height=12pt}, from=1-1, to=1-3]
        \arrow["{\overline{\normal}}", from=1-2, to=1-3]
    \end{tikzcd}
    \end{center}
    and it suffices to prove that $\overline{\normal}$ is injective (by uniqueness of factorisations, $\overline{\interpr-}$ is injective if and only if $\overline \normal$ is injective). In other words, we are to show that 
    $\ol{\normal}(\quotmap(a)) = \ol{\normal}(\quotmap(b))$ implies $a \quotrel b$. 
    We have:
    \begin{equation*}
        a \quotrel \normal(a) = \overline \normal (\quotmap(a)) = \overline\normal (\quotmap(b)) = \normal(b) \quotrel b,
    \end{equation*}
    where the first and the last equality follow from Normalisation (Proposition \ref{prop:normalisation}).
\end{proof}

As a corollary, we arrive at our main result, which justifies the title of the paper.

\begin{corollary}
\label{thm:thinImageAndQuotientIsomorphic}
    The $(F+G)$-algebras $(\thin Z, \thin \beta)$ and $(\Aquot, \alphaquot)$ are isomorphic. Hence thin behaviours form an initial coherent $(F+G)$-algebra.
\end{corollary}
\begin{proof}
    By Completeness of the quotient, we have an injective map $\overline{\interpr{-}}: (\Aquot, \alphaquot) \to (Z,\beta)$ with $\overline{\interpr-} \circ \quotmap = \interpr-$. Since $\thin Z = \Im(\interpr-)$, we have that $\overline{\interpr-}$ restricts to a surjective map from $(\Aquot, \alphaquot)$ onto $(\thin Z, \thin \beta)$. We conclude $(\Aquot, \alphaquot) \cong (\thin Z, \thin \beta)$, which by \Cref{prop:quotient} is an initial coherent $(F+G)$-algebra. The second part of the corollary follows from the coincidence of constructible and thin behaviours (\Cref{thm:thinIFFConstructible}).
\end{proof}

\Cref{thm:thinImageAndQuotientIsomorphic} allows us to define an algebraic notion of recognition of languages of constructible/thin elements of $Z$. Given a coherent $(F+G)$-algebra $(C, \gamma)$, by the initiality of $(\thin{Z},\thin{\beta})$ there is a unique $(F+G)$-algebra morphism $h \colon (\thin{Z},\thin{\beta}) \to (C, \gamma)$. We say that $(C, \gamma)$ \emph{recognises} a set 
$L \subseteq \thin{Z}$ of constructible behaviours if there is a $U \subseteq C$ such that $h^{-1}(U)=L$.
Trivially, $(Z,\beta)$ recognises any $L \subseteq \thin{Z}$ by taking $U=L$.
In future work, we aim to characterise those $L \subseteq \thin{Z}$ that are recognised by finite coherent algebras.

Finally, we give an example where \Cref{thm:thinImageAndQuotientIsomorphic} fails.

\begin{example}[\Cref{thm:thinImageAndQuotientIsomorphic} fails for $\Pow$]\label{ex:powerset-failure}
    \begin{figure}
        \centering
        \begin{subfigure}[b]{0.2\textwidth}
            \centering
            \scalebox{0.5}{
            \hspace{-0px}
            \begin{tikzpicture}
                \draw[streamLine] (0, -0.5) -- (0, 0);
                \draw[streamLine] (0, 0) -- (0, 1);
                \draw[streamLine] (0, 1) -- (0, 2);
                \draw[streamLine] (0, 2) -- (0, 3);
                \draw[streamLine] (0, 3) -- (0, 4);
                \filldraw[color=white] (0-0.1,0-0.1) rectangle ++(0.2,0.2);
                \draw[draw=blue] (0-0.1,0-0.1) rectangle ++(0.2,0.2);
                \filldraw[color=white] (0-0.1,1-0.1) rectangle ++(0.2,0.2);
                \draw[draw=blue] (0-0.1,1-0.1) rectangle ++(0.2,0.2);
                \filldraw[color=white] (0-0.1,2-0.1) rectangle ++(0.2,0.2);
                \draw[draw=blue] (0-0.1,2-0.1) rectangle ++(0.2,0.2);
                \filldraw[color=white] (0-0.1,3-0.1) rectangle ++(0.2,0.2);
                \draw[draw=blue] (0-0.1,3-0.1) rectangle ++(0.2,0.2);
                \node[stream] at (0, 4) {$G$};
                \node at (0, -1) {${\vdots}$};
            \end{tikzpicture}
            }
            \caption{$a_1$}
        \end{subfigure}
        \begin{subfigure}[b]{0.2\textwidth}
            \centering
            \scalebox{0.5}{
            \hspace{30px}
            \begin{tikzpicture}
                \draw[streamLine] (0, 0.5) -- (0, 1);
                \draw[streamLine] (0, 1) -- (0, 2);
                \draw[streamLine] (0, 2) -- (0, 3);
                \draw[streamLine] (0, 3) -- (0, 4);
                \node at (0, 0) {${\vdots}$};
                \node at (1, -0.5) {${\vdots}$};            

                \draw[dashed] (1, 3.5) -- (0, 4);
                \draw[streamLine] (1, 3.5) -- (2, 3);
                \draw[streamLine] (2, 3) -- (2.5, 2.75);
                \node at (3, 2.5) {$ \rotatebox[origin=c]{65}{\vdots}$};

                \draw[dashed] (1, 2.5) -- (0, 3);
                \draw[streamLine] (1, 2.5) -- (2, 2);
                \draw[streamLine] (2, 2) -- (2.5, 1.75);
                \node at (3, 1.5) {$ \rotatebox[origin=c]{65}{\vdots}$};

                \draw[dashed] (1, 1.5) -- (0, 2);
                \draw[streamLine] (1, 1.5) -- (2, 1);
                \draw[streamLine] (2, 1) -- (2.5, 0.75);
                \node at (3, 0.5) {$ \rotatebox[origin=c]{65}{\vdots}$};

                \draw[dashed] (1, 0.5) -- (0, 1);
                \draw[streamLine] (1, 0.5) -- (2, 0);
                \draw[streamLine] (2, 0) -- (2.5, -0.25);
                \node at (3, -0.5) {$ \rotatebox[origin=c]{65}{\vdots}$};

                \node[stream] at (0, 4) {$G$};
                \filldraw[color=white] (0-0.1,1-0.1) rectangle ++(0.2,0.2);
                \draw[draw=blue] (0-0.1,1-0.1) rectangle ++(0.2,0.2);
                \filldraw[color=white] (0-0.1,2-0.1) rectangle ++(0.2,0.2);
                \draw[draw=blue] (0-0.1,2-0.1) rectangle ++(0.2,0.2);
                \filldraw[color=white] (0-0.1,3-0.1) rectangle ++(0.2,0.2);
                \draw[draw=blue] (0-0.1,3-0.1) rectangle ++(0.2,0.2);

                \node[stream] at (1, 3.5) {$G$};
                \filldraw[color=white] (2-0.1,3-0.1) rectangle ++(0.2,0.2);
                \draw[draw=blue] (2-0.1,3-0.1) rectangle ++(0.2,0.2);

                \node[stream] at (1, 2.5) {$G$};
                \filldraw[color=white] (2-0.1,2-0.1) rectangle ++(0.2,0.2);
                \draw[draw=blue] (2-0.1,2-0.1) rectangle ++(0.2,0.2);

                \node[stream] at (1, 1.5) {$G$};
                \filldraw[color=white] (2-0.1,1-0.1) rectangle ++(0.2,0.2);
                \draw[draw=blue] (2-0.1,1-0.1) rectangle ++(0.2,0.2);

                \node[stream] at (1, 0.5) {$G$};
                \filldraw[color=white] (2-0.1,0-0.1) rectangle ++(0.2,0.2);
                \draw[draw=blue] (2-0.1,0-0.1) rectangle ++(0.2,0.2);
            \end{tikzpicture}
            }
            \vspace{7px}
            \caption{$a_2$}
        \end{subfigure}
        \caption{Terms $a_1$ and $a_2$, for $F = \Pow$.
        }
        \label{fig:powersetExample}
    \end{figure}

    Recall the finitary covariant power set functor $\Pow$. Just as for the bag functor $\Bag$, it is reasonable to define the type of one-hole contexts $\Pow' \coloneqq \Pow$. The plug-in is then defined as $\trig_X(Y, x) \coloneqq Y \cup \{ x \}$. One can then define the functor $G$, the initial $(F+G)$-algebra and the semantics $\interpr-$ in the same way as we did for analytic functors in \Cref{sec:alg-thin-rep}. However, we show that \Cref{thm:thinImageAndQuotientIsomorphic} \emph{does not} hold for $\Pow$. Consider the terms $a_1 \coloneqq \alpha_1((\emptyset)_{n \in \omega})$ and $a_2 \coloneqq \alpha_1((a_1)_{n \in \omega})$. They are represented graphically in \Cref{fig:powersetExample}. Notice that $\interpr{a_1} = \interpr{a_2}$, because $\Pow$ does not distinguish the number of successors.
    Consequently, if \Cref{thm:thinImageAndQuotientIsomorphic} was to hold, $a_1$ and $a_2$ would have to be identified in the initial coherent $(F+G)$-algebra. However, we disprove this by constructing a coherent $(F+G)$-algebra where $a_1$ and $a_2$ are not identified. Let $V$ be the set of terms $a \in A$ such that $a$ contains a nested subterm of the form $\alpha_1((\bar b'_n)_{n \in \omega})$ where $\bar b_n' \neq \emptyset$ for infinitely many $n \in \omega$ (here the ``nested subterm'' relation is the reflexive transitive closure of the subterm relation). For instance, one readily sees that $a_1 \notin V$, while $a_2 \in V$. Define the $(F+G)$-algebra $(C,\gamma)$:
    \begin{align*}
        & C = \{ V, A \setminus V \}, \qquad 
        \gamma_0(\bar c) = \begin{cases}
            V & \text{if $V \in \bar c$} \\
            A \setminus V & \text{otherwise},
        \end{cases} \\
        & \gamma_1((\bar c'_n)_{n \in \omega}) = \begin{cases}
            V & \text{if $V \in \bar c'_n$ for some $n$ or} \\
            & \hspace{10px} \text{$\bar c'_n \neq \emptyset$ for infinitely many $n$}, \\
            A \setminus V & \text{otherwise}.
        \end{cases}
    \end{align*}
    It is now straightforward to verify that $(C,\gamma)$ is coherent and that the unique $(F+G)$-morphism $\ev: (A,\alpha) \to (C,\gamma)$ satisfies $\ev(a) = V$ if and only if $a \in V$.
    Now observe that $A \setminus V = \ev(a_1) \neq \ev(a_2) = V$, hence, 
    since the quotient is a coequaliser, 
    $\quotmap(a_1) \neq \quotmap(a_2)$. Since $\interpr{a_1} = \interpr{a_2}$, this means 
    that $(\thin Z, \thin \beta)$ cannot be an initial coherent algebra.
\end{example}


\section{Connections to Descriptive Set Theory}
\label{sec:thin-trees-classic}

In this section, we connect our approach to specifying thin behaviours via terms in $A$
with the treatment of thin trees in descriptive set theory~\cite{kechris1995classical},
where thin trees are characterised in terms of the topological notions of Cantor-Bendixson derivative and rank.
We recall from \Cref{sec:preliminaries} that, for a \emph{polynomial} functor $F$, the elements of the final $F$-coalgebra can be seen as $F$-trees. In this case, both the major rank of the normal representative of an $F$-tree and the Cantor-Bendixson rank of a tree count the nesting level of the infinite branches. 
We prove in \Cref{thm:CB} that the two ranks coincide.

A minor technicality here is that, on a formal level, the definition of $F$-trees differs from the classic, topological definition of trees~\cite{kechris1995classical}. Classically, a tree is a prefix-closed subset of $\Sigma^*$ for some alphabet $\Sigma$. An $F$-tree can be translated to a classic tree by forgetting the $I$-labels.

For the rest of this section, we fix a polynomial functor $FX = \bigsqcup_{i \in I} X^{n_i}$, where $n_i$ are natural numbers, and use the notation $(i, (x_k)_{k \in n_i})\in FX$ and $(i, j, (x_k)_{k \neq j}) \in F'X$.

\subsection{Preliminaries on Trees}

We first recall the relevant definitions from descriptive set theory and fix notation.
For more details, see \cite{kechris1995classical,Skrzypczak2016,GoyThinTrees}.

Let $\Sigma$ be an arbitrary set called the alphabet.
For a finite or infinite word $w \in \Sigma^* \cup \Sigma^\omega$,
$\Pref{w}$ denotes the set of prefixes of $w$:
$\Pref{w} = \{ u \in \Sigma^* \mid \exists v \in  \Sigma^* \cup \Sigma^\omega: w = uv \}$.
For $L \subseteq \Sigma^* \cup \Sigma^\omega$, we define
$\Pref{L} = \bigcup_{w \in L} \Pref{w}$.
Furthermore, for $u \in \Sigma^*, uL = \{ uv \mid v \in L \}$ and
$u^{-1}L = \{ v \in \Sigma^* \cup \Sigma^\omega \mid uv \in L \}$.
For $w \in \Sigma^\omega$ and $n \in \omega$, the prefix of $w$ of length $n$ is denoted 
$\restr{w}{n}$.

A tree $\tree$ over $\Sigma$ is a prefix-closed subset of $\Sigma^*$, i.e., 
$\Pref{\tree}\subseteq \tree$.
We denote with $\Trees$ the set of all trees over $\Sigma$.
For a tree $\tree\in \Trees$, we define the set of infinite branches of $\tree$ as 
$\infbr{\tree}= \{ w \in \Sigma^\omega \mid \Pref{w} \subseteq \tree\}$.    

For $\tree\in \Trees$ and $u \in \tree$, note that
$u^{-1}\tree$ is the subtree of $\tree$ rooted at $u$,
$uu^{-1}\tree$ is the subset of $\tree$ (generally not a tree) 
consisting of words in $\tree$ with prefix $u$,
and
$\Pref{u} \cup uu^{-1}\tree= \Pref{uu^{-1}\tree}$ is the tree obtained by restricting to nodes along $u$ 
and the subtree $u^{-1}\tree$. 

Let $\tree\in \Trees$. We say that an infinite branch $w \in \infbr{\tree}$ is \emph{isolated} if there exists $n \in \omega$ s.t. 
$\restr{w}{n}\Sigma^\omega \cap \infbr{\tree}= \{w\}$. 
Informally, $w$ is isolated if for some $n \in \omega$, the subtree of $\tree$ rooted at $\restr{w}{n}$ is non-branching, i.e., each node has exactly one child.
The \emph{Cantor-Bendixson derivative} (CB-derivative or derivative, for short) of $\infbr{\tree}$ is the subset $\infbr{\tree}' \subseteq \infbr{\tree}$ defined as
\begin{equation}
    \infbr{\tree}' = \infbr{\tree}\setminus\{ w \in \infbr{\tree}\mid w \text{ is isolated} \}
\end{equation}
The derivative can be iterated.
\begin{align*}
    \derivIter{\infbr{\tree}}{0} &= \infbr{\tree},  \\
    \derivIter{\infbr{\tree}}{\ord+1} &= ({\derivIter{\infbr{\tree}}{\ord}})', \\
    \derivIter{\infbr{\tree}}{\lambda} &= \bigcap_{\ord < \lambda} \derivIter{\infbr{\tree}}{\ord} \quad \text{$\lambda$ is limit ordinal}.
\end{align*}
\begin{remark}
    The above definition of the CB-derivative is equivalent to the one from topology~\cite{kechris1995classical} using the topological space on $\Sigma^\omega$ where  $\{ u\Sigma^\omega \mid u \in \Sigma^*\}$ (``cylinder sets'') is the basis. With this topology, isolated infinite branches in $\infbr{\tree} \subseteq \Sigma^\omega $ are precisely those that are isolated in the subspace topology on $\infbr{\tree}$ inherited from $\Sigma^\omega$.
\end{remark}

The sequence $\derivIter{\infbr{\tree}}{\ord}$ for all ordinals $\ord$ is decreasing and hence stabilises.
The \emph{Cantor-Bendixson rank (CB-rank)} of a tree $\tree$, denoted $\CB{\tree}$, is defined as the least ordinal $\ord$ such that $\derivIter{\infbr{\tree}}{\ord}=\derivIter{\infbr{\tree}}{\ord+1}$.
We can also define a derivative on trees in $T(\Sigma)$ by taking
$\tree' = \Pref{\infbr{\tree}'}$. It follows that for all ordinals $\ord$, $\derivIter{\infbr{\tree}}{\ord} = \infbr{\derivIter{\tree}{\ord}}$, i.e., the derivative operation commutes with taking infinite branches.
We observe that for a finitely branching tree $\tree$, the derivative of $t$ is obtained by
removing from $\tree$ all nodes $u$ 
such that $u^{-1}\tree$ has finitely many (finite or infinite) branches.

For trees $\tree$ over a countable alphabet $\Sigma$, we have the classic result that links CB-rank 
and having countably many infinite branches, cf.~\cite{kechris1995classical,Skrzypczak2016}. 
\begin{lemma}\label{lem:thin-countably}
For a countable alphabet $\Sigma$ and $\tree\in T(\Sigma)$, we have 
$\derivIter{\infbr{\tree}}{\CB{{\tree}}} = \emptyset$
iff
$\infbr{\tree}$ is countable.
\end{lemma}
Given a countable $\Sigma$, a tree $\tree\in T(\Sigma)$ is \emph{thin} if it satisfies any of the two equivalent conditions from \Cref{lem:thin-countably}.

An $F$-tree $\tau \in Z$ can be formally modelled as a function $\tau : t \to I$, where $t \in T(\omega)$ and for all $u \in t$, if $\tau(u) = i$ then $u$ has exactly $n_i$ children $u0, \ldots, u(n_i-1)$.

\subsection{Encoding and Rank}

In order to obtain a set-theoretic description of terms as trees,
we define an encoding of terms from $A$ in $T(\omega)$ by induction, i.e., by using initiality of $A$.

\begin{definition}\label{def:alg-trees}
    We define an $(F+G)$-algebra structure $\phi=[\phi_0,\phi_1] \colon (F+G)(T(\omega)) \to T(\omega)$ as follows:
    \begin{align*}
        & \phi_0\colon FT(\omega) \to T(\omega), \quad  (i,(\tree_k)_{k \in n_i}) \mapsto \{\epsilon\} \cup \bigcup\limits_{k \in n_i} k \tree_k \\
        & \phi_1 \colon G T(\omega) \to T(\omega), \quad (C_i)_{i \in \omega} \mapsto \\
        & \hspace{110px} \bigcup\limits_{n \in \omega} \restr{(d_{T(\omega)})^\omega((C_i)_{i \in \omega}))}{n} p(C_n)
    \end{align*}
where $d_{T(\omega)}$ is a component of the natural transformation
$d_X \colon F'X \to \omega$ given by $d_X(i,k, - )=k$,  
and $p : F' T(\omega) \to T(\omega)$ maps $(i,j,(\tree_k)_{k \ne j})$ to $\{\epsilon\} \cup \bigcup_{k \ne j}k \tree_k$.
(Note that the codomain of $d$ is the constant functor with value $\omega$.)

We denote by $\encA: (A,\alpha) \to (T(\omega),\phi)$ the unique $(F+G)$-algebra morphism obtained from initiality of $(A,\alpha)$. 
\end{definition}
The intuition here is that $d$ extracts from a context the direction where the hole is located, and thus 
$(d_{T(\omega)})^\omega((C_i)_{i \in \omega})$ extracts the required infinite branch $w$ along which the $\omega$-tree $p(C_n)$ is glued at position $\restr{w}{n}$.

Next we show that $\encA$ factors via the final $F$-coalgebra.

\begin{toappendix}
    We will use the coinduction principle obtained from the final $F$-coalgebra.
    We recall the basic definitions, and refer to \cite{Rutten:TCS2000} for more details.
    First, the abstract definition of $F$-bisimulation \cite{Rutten:TCS2000} instantiates as follows for polynomial functors $F$.
    \begin{definition}[$F$-bisimulation]
    Given an $F$-coalgebra $(X,\xi)$, a relation $R \subseteq X \times X$ is an \emph{$F$-bisimulation}
    if for all $(x,y) \in R$ where $\xi(x) = (i_x, x_0, \ldots, x_{{n_i}-1})$ and  $\xi(y) = (i_y, y_0, \ldots, y_{{n_i}-1})$,
    we have that $i_x = i_y =:i$ and for all $j \in n_i$, $(x_j,y_j) \in R$.
    \end{definition}

    \begin{definition}[$F$-coinduction]
        If $(Z,\zeta)$ is a final $F$-coalgebra and $R \subseteq Z \times Z$ is an $F$-bisimulation
        then for all $(z_2,z_2) \in R$, $z_1=z_2$.
    \end{definition}    

We use the  next lemma to prove \Cref{lem:encA-dom}
    \begin{apxlemmarep}\label{lem:inf-branch-decomp}
        Let  $\bar{z} = (\bar{z}'_j)_{j \in \omega} \in GZ = (F'Z)^\omega$, 
        let $\tau := \beta_1(\bar{z})$ and let $w := (d_Z)^\omega(\bar{z}) \in \omega^\omega$ 
        be the infinite branch of $\tau$ encoded in $\bar{z}$.
        For all $m \in \omega$, $\beta_1( (\bar{z}'_j)_{j \geq m} ) = (\restr{w}{m})^{-1}\tau$.
    \end{apxlemmarep}
    \begin{proof}
        We prove that the relation
        $R := \Delta_Z \cup \{ (\beta_1( (\bar{z}'_j)_{j \geq m} ) , (\restr{w}{m})^{-1}\tau) \mid m \in \omega\}$
        is an $F$-bisimulation. The result then follows by $F$-coinduction.
        For $(z_1,z_2) \in R$, the $F$-bisimulation condition holds trivially.
        For $(\beta_1( (\bar{z}'_j)_{j \geq m} ) , (\restr{w}{m})^{-1}\tau) \in R$, we compute as follows.
        By the definition of $w$, we have:
        \[
        \zeta((\restr{w}{m})^{-1}\tau)  =  \trig_Z(\bar{z}'_m, (\restr{w}{m+1})^{-1}\tau )
        \]
        On the other hand, since $[\beta_0,\beta_1]$ is coherent, we have:
        \[
        \zeta(\beta_1( (\bar{z}'_j)_{j \geq m} )) = \trig_Z(\bar{z}'_m, \beta_1( (\bar{z}'_j)_{j \geq m+1}))
        \]
        Since $( \beta_1( (\bar{z}'_j)_{j \geq m+1} , (\restr{w}{m+1})^{-1}\tau) \in R$ and the other subtrees are related by identity,
        the $F$-bisimulation condition holds.
    \end{proof}
\end{toappendix}

\begin{apxlemmarep}\label{lem:encA-dom}
Let $(Z,\zeta)$ be the final $F$-coalgebra of all $F$-trees.
Let $\dom\colon Z \to T(\omega)$ be the map that sends $\tau\colon \tree\to I$ to $\dom(\tau)=\tree$.
We have: $\encA = \dom \circ \interpr{-}$. 
\end{apxlemmarep}
\begin{proof}
It suffices to prove that $\dom$ is an $(F+G)$-algebra morphism from $(Z,\beta)$ to $(T(\omega),\phi)$,
since then $\dom \circ \interpr{-}$ is also an $(F+G)$-algebra morphism, hence by initiality of $(A,\alpha)$,
it follows that  $\encA = \dom \circ \interpr{-}$. 

To prove that $\dom$ is an $(F+G)$-algebra morphism from $(Z,\beta)$ to $(T(\omega),\phi)$,
we must prove:
\begin{enumerate}
\item For all $\bar{z} \in FZ$, $\phi_0(F \dom(\bar{z})) = \dom(\beta_0(\bar{z}))$.
\item For all $\bar{z} \in GZ$, $\phi_1(G \dom(\bar{z})) = \dom(\beta_1(\bar{z}))$
\end{enumerate}
To prove item 1, first note that since $\zeta$ is a bijection, 
all elements of $FZ$ are of the form $\zeta(\tau)$ for some $\tau \in Z$.
So let $\zeta(\tau) \in FZ$ and let $(i, (\tau_j)_{j \in n_i}) := \zeta(\tau)$.
That is, the root of $\tau$ is labelled $i$, and the children of $\tau$
are $(\tau_j)_{j \in n_i}$.
Then we have,
$\phi_0(F\dom(\zeta(\tau))) = \{\emptyword\} \cup \bigcup_{j \in n_i} j \dom(\tau_j) = \dom(\tau) = \dom(\beta_0(\zeta(\tau)))$,
where the last identity holds since $\beta_0=\zeta^{-1}$.

To prove item 2, let $\bar{z} = (\bar{z}'_j)_{j \in \omega} \in GZ$.
By definition of $\phi_1$,
\begin{multline*}
    \phi_1(G \dom(\bar{z})) = \phi_1((F'\dom)^\omega(\bar{z})) = \\
    \bigcup_{n \in \omega} \restr{((d_{T(\omega)} \circ F'\dom)^\omega(\bar{z}))}{n} p(F'\dom(\bar{z}'_n)).
\end{multline*}
First note that by naturality of $d$, $d_{T(\omega)} \circ F'\dom = d_Z$.
Let $\tau := \beta_1(\bar{z})$. Then
$w := (d_{T(\omega)} \circ F'\dom)^\omega(\bar{z})$ is the infinite branch of $\tau$ encoded in $\bar{z}$.
We can decompose $\tree:= \dom(\tau)$ along $w$:
\[
 \tree= \Pref{w} \cup \bigcup_{n \in \omega} (\restr{w}{n})(\restr{w}{n})^{-1}T.
\]
We must then show that
\[
 \bigcup_{n \in \omega} (\restr{w}{n}) p(F'\dom(\bar{z}'_n)) = \Pref{w} \cup \bigcup_{n \in \omega} (\restr{w}{n})(\restr{w}{n})^{-1}T.
\]

$(\subseteq)$: Let $u = (\restr{w}{n})v$ where $n \in \omega$ and $v \in p(F'\dom(\bar{z}'_n))$. 
Assume that $\bar{z}'_n = (i,k, (\tau_j)_{j \in n_i\setminus \{k\}})$.
Then $F'\dom(\bar{z}'_n) = (i,k, (\dom(\tau_j))_{j \in n_i\setminus \{k\}})$.
If $v = \emptyword$ then $u = \restr{w}{n}  \in \Pref{w}$, hence
$u \in \dom(\tau)$.
If $v = j \dom(\tau_j)$ for some $j \in n_i\setminus \{k\}$ then
$u \in  (\restr{w}{n}) j \dom(\tau_j)$.
From \Cref{lem:inf-branch-decomp}, it follows that $\tau_j$ is the subtree of $\tau$ rooted at $ (\restr{w}{n})j$.
Hence $u \in \dom(\tau)$.

$(\supseteq)$:
If $u \in \Pref{w}$, then $u = \restr{w}{n}$ for some $n \in \omega$, and since $\emptyword \in p(F'\dom(\bar{z}'_n))$ for all $n \in \omega$,
it follows that $u \in (\restr{w}{n}) p(F'\dom(\bar{z}'_n))$.
Assume now that 
$u \in (\restr{w}{n})(\restr{w}{n})^{-1}\tree$ for $n \in \omega$,
i.e., $u = (\restr{w}{n})v$ for some $v \in (\restr{w}{n})^{-1}\tree$.
Assume further that $n$ the maximal, i.e., for all $m > n$,  $u \notin (\restr{w}{m})(\restr{w}{m})^{-1}\tree$,
and that $\bar{z}'_n = (i,k, (\tau_j)_{j \in n_i\setminus \{k\}})$. 
We must show that for some $j \in n_i\setminus \{k\}$, $v \in j \dom(\tau_j)$. 
By \Cref{lem:inf-branch-decomp}, 
$(\restr{w}{n})^{-1}\tree= \dom(\beta_1(\bar{z}'_j)_{j \geq n}) = \dom(\trig_Z(\bar{z}'_n, \beta_1( (\bar{z}'_j)_{j \geq n+1})))$.
Since $n$ is maximal, by \Cref{lem:inf-branch-decomp}, $v \notin ({w}_{n+1})\dom(\beta_1((\bar{z}'_j)_{j \geq n+1}))$.
It follows that for some $j \in n_i\setminus \{k\}$, $v \in j \dom(\tau_j)$. Hence $v \in p(F'\dom(\bar{z}'_n))$.
\end{proof}

We can now state and prove the result that connects the major rank of a normal term with the Cantor-Bendixson rank of the associated tree. This result further motivates our choice of normal terms as a natural one. (Note that Theorem~\ref{thm:CB} does not hold if $a$ is not normal.) 

\begin{toappendix}
    \begin{lemma}
    \label{lem:CBonF'}
    For $C \in F' T(\omega)$, $\CB{p(C)} = \sup\{\CB{\tree} \mid \tree\in \Base_{F'}(C)\}$.
    \end{lemma}
\end{toappendix}

\begin{apxtheoremrep}\label{thm:CB}
For all $a \in \Normal$, $\infbr{\encA(a)}$ is thin and $\CB{\encA(a)} = \majrk(a)$.
\end{apxtheoremrep}
\begin{proofsketch}
We prove the statement by induction on the structure of $a \in \Normal$. When $a = \alpha_0(i, (a_k)_{k \in \{0, \dotsc, n_i - 1\}}) \in \Branch$, we have that an infinite branch is isolated in $[\encA(a)]$ iff its suffix is isolated in the respective $[\encA(a_k)]$. Hence, using the induction hypothesis, we obtain:
\begin{multline*}
    \CB{\encA(a)} = \sup\{\CB{\encA(a_k)} \mid k \in \{0,\ldots,n_i-1\} \} \\
    = \sup \{ \majrk(a_k) \mid k \in \{0,\dotsc,n_i - 1\}\} = \majrk(a).
\end{multline*}
When $a = \alpha_1((\bar a_n')_{n\in\omega})$, we have $\majrk (a) = \alpha + 1 > \majrk(b)$ for all $b \in \Base_{F'}(\bar a_n')$ and $n \in \omega$. Writing $w = (d_A)^\omega(a)$ for the main branch of $\encA(a)$ and $C_n := (F' \encA)(\bar a_n')$,
we show that 
  \begin{equation*}
        [\encA(a)] = \{w\} \cup \bigcup\limits_{n \in \omega} \restr{w}{n} [p(C_n)].
        \end{equation*}
Using $[p(C_n)]^{(\alpha)} = \emptyset$ (as $\CB{p(C_n)} = \sup \{ \majrk(b) \mid b \in \Base_{F'}(\bar a_n') \} < \alpha + 1$), we obtain $[\encA]^{(\alpha)} = \{w\}$, so $[\encA(a)]^{(\alpha+1)} = \emptyset$ and $\CB{[\encA(a)]} = \alpha+1$.
\end{proofsketch}
\begin{proof}
    We prove the statement by induction on the structure of $a \in \Normal$.
    \begin{enumerate}
        \item $a \in \Branch$. We prove the statement by induction on $\minrk(a)$. Say $a = \alpha_0(i,(a_k)_{k \in \{0,\ldots,n_i-1\}})$, with $a_k$ normal for all $k$, and assume that the statement holds for all $a_k$.
        Then, by the definition of major rank and using the induction hypothesis, we have:
        \begin{eqnarray*}
            \majrk(a) & = & \sup \{\majrk(a_k) \mid k \in \{0,\ldots,n_i-1\}\} \\
            & = &\sup \{\CB{\encA(a_k)} \mid k \in \{0,\ldots,n_i-1\}\}
        \end{eqnarray*}
        On the other hand, writing $\tree_k = \encA(a_k)$, 
        we have:
        \begin{eqnarray*}
            [\encA(a)] & = & [\{\epsilon\} \cup \bigcup\limits_{k \in \{0,\ldots,n_i-1\}} k \tree_k]\\
            & = & \bigcup\limits_{k \in \{0,\ldots,n_i-1\}} k [\tree_k]
        \end{eqnarray*}
        As the sets of infinite branches in the above union are pairwise disjoint, a branch is isolated in $[\encA(a)]$ if and only if it is isolated in the respective $k [\tree_k]$. As a result, we have
        \begin{multline*}
            \CB{\encA(a)} =  \sup\{\CB{\tree_k} \mid k \in \{0,\ldots,n_i-1\} \} \\
            =  \sup\{\CB{\encA(a_k)} \mid k \in \{0,\ldots,n_i-1\} \}.
        \end{multline*}
        We therefore obtain $\majrk(a) = \CB{\encA(a)}$. We also have
        \begin{multline*}
            [\encA(a)]^{(\CB{\encA(a)})} = \\ 
            \bigcup\limits_{k \in \{0,\ldots,n_i-1\}} k[\tree_k]^{(\CB{\encA(a)})}  = \emptyset
        \end{multline*}
        since $\CB{\tree_k} = \CB{\encA(a_k)} \le \CB{a}$. This concludes the proof in the case $a \in \Branch$.
        \item $a \in \Stream$. Then necessarily $\majrk(a) = \alpha+1$, with $\alpha$  either a limit or a successor ordinal. Say $a = \alpha_1((\bar x_n')_{n \in \omega})$ with $\bar x_n' \in F' A$ and $\alpha+1 > \majrk(b)$ for each $b \in \Base_{F'}(\bar x_n')$.
        Let $w = (d_A)^\omega(a)$ be the main branch of $\encA(a)$, and $C_n := (F' \encA)(\bar x_n') \in F' T(\omega)$ for $n \in \omega$. Thus, $p(C_n)$ is the $\omega$-tree that is attached to $w$ at depth $n$ in $\encA(a)$. Since $\majrk(a) = \alpha+1$, it follows from the definition of major rank on $G$-terms (\Cref{def:majorAndMinorRank}) that
        for all \emph{successor} ordinals $\beta < \alpha + 1$, for infinitely-many $n \in \omega$, there exists a normal term $a_n \in \Base_{F'}(\bar x_n')$
        with $\alpha+1 > \majrk(a_n) \ge \beta$. Now since $\majrk(a_n) < \alpha+1$, the inductive hypothesis applies to each $a_n$, and therefore $\alpha+1 > \CB{\encA(a_n)} \ge \beta$. Then, since $\encA(a_n) \subseteq p(C_n)$, by monotonicity
        of the Cantor-Bendixson derivative, for all successor ordinals $\beta < \alpha+1$, for infinitely-many $n \in \omega$, $\CB{p(C_n)} \ge \beta$. Moreover, we also have $\alpha+1 > \majrk(b) = \CB{\encA(b)}$ for each $b \in \Base_{F'}(\bar x_n')$, and therefore, by Lemma~\ref{lem:CBonF'}, $\alpha+1 > \CB{p(C_n)}$ (since $p(C_n) = \{\epsilon \} \cup \bigcup\limits_{b \in \Base_{F'}(\bar x_n')}k_b\encA(b)$ for suitable, pairwise-different choices of $k_b$ with $b \in \Base_{F'}(\bar x_n')$, and thus $\CB{p(C_n)} = \sup\{\CB{\encA(b)} \mid b \in \Base_{F'}(\bar x_n')\}$.).

        Now recall that $\encA(a)$ is given by $\bigcup\limits_{n \in \omega}\restr{w}{n} (\{\epsilon\} \cup p(C_n))$. Note that this set always includes $w$. We then have:
    \begin{equation}
    \label{eqn:enca}
        [\encA(a)] = \{w\} \cup \bigcup\limits_{n \in \omega} \restr{w}{n} [p(C_n)]
    \end{equation}
    Since the above unions are disjoint, an infinite branch is isolated in $[\encA(a)]$ if and only if it is isolated in the corresponding $\restr{w}{n} [p(C_n)]$. On the other hand, we have that $w \in [\encA(a)]^{(\alpha)}$.
    This is shown by distinguishing two cases: \begin{itemize}
        \item if $\alpha$ is a limit ordinal, then $w \in [\encA(a)]^{(\alpha)}$ follows from $w \in [\encA(a)]^{(\beta)}$ for all successor ordinals $\beta < \alpha$, which in turn follows from $\CB{p(C_n)} \ge \beta$ for infinitely many $n$, for all successor ordinals $\beta < \alpha$.
        \item if $\alpha$ is a successor ordinal, say $\alpha = \beta + 1$, then $w \in [\encA(a)]^{(\alpha)}$ follows from $\CB{p(C_n)} \ge \beta$ for infinitely many $n$.
    \end{itemize}
    Now applying the Cantor-Bendixson derivative $\alpha$ times to equation (\ref{eqn:enca}) we obtain:
        \begin{equation*}
        [\encA(a)]^{(\alpha)} = \{w\} \cup \bigcup\limits_{n \in \omega} \restr{w}{n}[p(C_n)]^{(\alpha)}
    \end{equation*}
    Since each $\CB{p(C_n)} = \majrk(a_n) < \alpha + 1$, we have $[p(C_n)]^{(\alpha)} = \emptyset$ for all $n \in \omega$. We thus obtain (again using $w \in [\encA(a)]^{(\alpha)}$) \begin{equation*}
    [\encA(a)]^{(\alpha)} = \{w\}    
    \end{equation*}
    As a result, $[\encA(a)]^{(\alpha+1)} = \emptyset$ and therefore $\CB{\encA(a)} = \alpha+1$. This concludes the proof also for the case $a \in \Stream$.
    \end{enumerate}
\end{proof}

As a result, in the case of polynomial functors, our notion of rank of a thin element (\Cref{def:thinRankCoalg}) recovers the Cantor-Bendixson rank of the associated standard tree.
\begin{corollary}
\label{cor:ranks}
    For $\tau \in Z$ with $\rk(\tau) = (r_{\mathrm{major}}, r_{\mathrm{minor}})$, we have $r_{\mathrm{major}} = \CB{dom(\tau)}$.
\end{corollary}


\section{Conclusion}\label{sec:conclusion}

In this paper, we introduced the notion of a thin $F$-coalgebra for an analytic set functor $F$. To do this, we
introduced the notion of a path through a coalgebra for an analytic functor and defined thin coalgebras as those whose states give rise to at most countably many paths. As a first result we obtained a combinatorial characterisation of thin coalgebras in terms of counting cycles, which is verifiable in linear time.


We proceeded with a syntactic characterisation of thin coalgebra behaviours. We defined a syntax for thin behaviours as an initial $(F+G)$-algebra and equipped this syntax with equivalent operational and denotational semantics. We called elements in the image of the semantics map constructible behaviours and showed that they coincide with thin elements of the final $F$-coalgebra. Furthermore, we established an inductive structure of thin behaviours by giving a sound and complete axiomatisation of the semantics.   
The completeness proof relied on showing that each term can be converted (using \Cref{eq:plugging})
to a unique normal term with the same semantics. These normal terms also enabled us to syntactically measure the rank of a thin coalgebra element.

To connect with existing work on thin trees, we instantiated our framework to polynomial functors.
In this case, the behaviour of an $F$-coalgebra is an $F$-tree,
the behaviour of a thin $F$-coalgebra is a thin $F$-tree, and thus thin $F$-trees are the semantics of terms in the initial $(F+G)$-algebra. Furthermore, we have shown that in this case, our notion of rank of an $F$-tree captures the Cantor-Bendixson rank of trees \cite{Skrzypczak2016}. Thus, similarly to the Cantor-Bendixson rank, our rank can be seen as a way to measure the degree of thinness of a thin $F$-tree, and by extension of a state in a thin $F$-coalgebra.

Regular languages of thin trees have the remarkable property of being recognised by unambiguous automata. One central aim for future work will be to lift this result to thin $F$-coalgebras and $F$-coalgebra automata~\cite{VenemaKupke:CoalgebraicAutomata}. This will be important for the work in~\cite{CirsteaKupke:CSL2023} which outlines how unambiguous automata can be used to verify quantitative (fixpoint) properties of state-based systems modelled as coalgebras. To achieve this we also plan to further develop the theory of (coherent) $(F+G)$-algebras, similar to existing work on Wilke algebras \cite[Section~2.5]{PerrinPin:2004:InfiniteWords} and thin algebras~\cite{BojIdzSkr:STACS13,Skrzypczak2016}. 

Another important direction will be to study the behaviour of regular thin $F$-coalgebras, i.e., of thin $F$-coalgebras that have only finitely many states. This is because the language of an automaton is characterised by the accepted regular $F$-coalgebras. We conjecture that regular thin $F$-coalgebras correspond precisely to the finitary $(F+G)$-terms, i.e., terms where all $G$-streams are ultimately periodic. This would enable us to reduce our infinitary syntax to a finitary one, thus allowing for practical applications of this syntax as a specification language.
We also plan to explore whether our insights into the algebraic representation of thin $F$-coalgebras can be used to generalise $\Omega$-automata~\cite{CianciaVenema2019OmegaAutomataACoalgebraicPerspective} from infinite words to coalgebras.

Finally, our syntactic characterisation of thin $F$-coalgebras (\Cref{thm:thinIFFConstructible}) is categorical and paves the way for generalisations beyond analytic set functors: Given an arbitrary endofunctor $F$ with a well-behaved notion of functor derivative, if the initial $(F+G)$-algebra and the final $F$-coalgebra exist, one can define constructible $F$-behaviours and study their properties. 
We plan to study the limits of this approach on $\Set$,
and to look into generalised analytic functors on categories beyond $\Set$, e.g.~\cite{FioreAnalyticFunctorPresheaf, kock2012}.
On $\Set$, we suspect that analytic functors are, in fact, the limit, since any further quotienting on $F$ would likely destroy the initiality of thin behaviours and invariance of the number of paths under morphisms.
We exemplified this in \Cref{ex:morphismPreservationFailure,ex:powerset-failure} for the finitary covariant power set functor, 
a non-analytic functor that can be obtained by quotienting the bag functor with idempotence.  

\textbf{Acknowledgements.} We are grateful to Alexandre Goy for his involvement in earlier versions of this paper. Our discussions with him, together with his notes on infinite trees \cite{GoyThinTrees}, helped us in formulating and proving the results in \Cref{sec:thin-trees-classic}. We would also like to thank Jade Master for numerous helpful discussions on thin trees.

%
%
\printbibliography

\end{document}